\documentclass[12pt,reqno]{amsart}
\usepackage{amsmath,amscd,amsfonts,amssymb, color,bbm, bm,
braket, mathtools}
\usepackage{latexsym, graphicx, pstricks,rotating, enumerate}
\usepackage[margin=1in]{geometry}
\usepackage[colorlinks = true, linkcolor = blue, urlcolor  = blue, citecolor = red,pagebackref]{hyperref}
\usepackage[capitalise]{cleveref}
\numberwithin{equation}{section}
\newtheorem{prop}{Proposition}[section]
\newtheorem{defin}{Definition}[section]
\newtheorem*{defin*}{Definition}
\newtheorem{rem}{Remark}[section]

\newtheorem{theorem}{Theorem}[section]

\newtheorem{innercustomthm}{Theorem}
\newenvironment{customthm}[1]
  {\renewcommand\theinnercustomthm{#1}\innercustomthm}
  {\endinnercustomthm}
\newcommand{\tr}{\mathrm{Tr\,}}
\newcommand{\dd}{\mathrm{d}}

\usepackage{mathrsfs}
\usepackage{float}

\definecolor{MathematicaBlue}{rgb}{0.368417, 0.506779, 0.709798}
\definecolor{MathematicaOrange}{rgb}{0.880722, 0.611041, 0.142051}

\newcommand{\ii}{\mathrm{i}}
\newcommand{\GOE}{\mathrm{GOE}}
\newcommand{\RQCM}{\mathrm{RQCM}}

\renewcommand{\epsilon}{\varepsilon}

\renewcommand{\O}{\mathsf{O}}
\newcommand{\Sp}{\mathsf{Sp}}
\newcommand{\K}{\mathsf{K}}

\begin{document}
\title{Generating random Gaussian states}
\author{Leevi Lepp\"aj\"arvi}
\address{RCQI, Institute of Physics, Slovak Academy of Sciences,
D\'ubravsk\'a cesta 9, 84511 Bratislava, Slovakia}
\email[Leevi Lepp\"aj\"arvi]{\href{mailto:leevi.leppajarvi@savba.sk}{leevi.leppajarvi@savba.sk}}
\author{Ion Nechita}
\address{Laboratoire de Physique Th\'eorique, Universit\'e de Toulouse, CNRS, UPS, France}
\email[Ion Nechita]{\href{mailto:ion.nechita@univ-tlse3.fr}{ion.nechita@univ-tlse3.fr}}
\author{Ritabrata Sengupta}
\address{Department of Mathematical Sciences, Indian Institute of
Science Education \& Research (IISER) Berhampur, Transit campus,
Berhampur 760 010, Ganjam, Odisha, India}
\email[Ritabrata Sengupta]{\href{mailto:rb@iiserbpr.ac.in}{rb@iiserbpr.ac.in}}
\begin{abstract}
We develop a method for the random sampling of (multimode) Gaussian states in terms of their covariance matrix, which we refer to as a random quantum covariance matrix (RQCM). We analyze the distribution of marginals and demonstrate that the eigenvalues of an RQCM converge to a shifted semicircular distribution in the limit of a large number of modes. We provide insights into the entanglement of such states based on the positive partial transpose (PPT) criteria. Additionally, we show that the symplectic eigenvalues of an RQCM converge to a probability distribution that can be characterized using free probability. We present numerical estimates for the probability of a RQCM being separable and, if not, its extendibility degree, for various parameter values and mode bipartitions.

\vspace{1cm}

\noindent \textbf{Keywords.} Gaussian state, random matrices,  extendability, entanglement, PPT criterion, symplectic eigenvalues.  

\smallskip

\noindent \textbf{Mathematics Subject Classification
(2010):} 81P40, 81P99, 94A15.
\end{abstract}

\maketitle

\tableofcontents

\newpage

\section{Introduction}

Quantum mechanics can be considered as a non-commutative version of classical probability theory \cite{mackey, krp4a, meyer, krp7}. However most often in quantum information theory, probability theory was used in the areas which rely on the measurement statistics. The trend changed with the seminal paper of Hayden et al \cite{hayden2006aspects} which used random matrix and concentration of measure techniques in quantum information theory.  Henceforth, random matrix techniques have been used extensively in this area. For construction of random states, see the paper of \.{Z}yczkowski et al \cite{zyczkowski2011generating}, and the review article \cite{collins2016random}. 
\par Continuous variable quantum systems are one of the most useful systems for the quantum information processing. Among these, Gaussian states play a crucial role. Gaussian states and their applications are studied extensively in the literature. These concepts, explained from a mathematical point can been seen in the book of Holevo  \cite{MR2797301}. Gaussian  states are easy to prepare, manipulate, or measure  and have wide range of applications in quantum optics and quantum communications \cite{RevModPhys.84.621, adesso_osid}.  
\par There are very few works on using random matrix techniques and Gaussian states. The only attempt as per our knowledge if the work of  Serafini et al \cite{MR2345312} which was further developed by Fukuda and K\"onig \cite{fukuda-koenig-typ}. In these papers, the authors introduced a method to generate the random Gaussian states. The novelty of the present work is that it uses a different method of sampling. Furthermore, this also explores the problem of entanglement detection for such states.  

Gaussian states can be represented by their covariance matrices which are real positive definite matrices satisfying a Heisenberg-type uncertainty relation. This is the point of view we adopt in this work for sampling Gaussian states. We are going to modify a well-known ensemble of real random matrices (the Gaussian Orthogonal Ensemble - GOE) in order to impose the uncertainty condition. 

Recall that an element of the GOE is a real symmetric random matrix, with Gaussian entries that are independent (up to the symmetry condition). Wigner \cite{wigner1955characteristic} famously showed that the eigenvalues of GOE random matrices converge to the semicircle distribution (see Figure \ref{fig:GOE-histogram}). In particular, this means that GOE matrices are likely to have negative eigenvalues, so they cannot be covariance matrices of Gaussian quantum states\footnote{Let us point out that the word ``Gaussian'' is used above in two different contexts. On the one hand, in quantum theory, an $n$-mode Gaussian state is a quantum state on $L^2(\mathbb R^n)$ with the property that its expectation with respect to position and momentum observables on different modes are Gaussian random variables. On the other hand, in random matrix theory, the Gaussian Orthogonal Ensemble is a probability distribution on the set of symmetric real matrices which has independent, identically distributed entries (up to symmetry).}. In order to fulfill the uncertainty principle, we shall uniformly shift the GOE element by the minimal amount needed to satisfy the uncertainty principle. We introduce in this way the ensemble of random matrices that we shall study: 
\begin{defin*}
    A \emph{random quantum covariance matrix} (RQCM) is defined by
    $$S_G:=G + \lambda_{\max}(\ii J_{2n} - G) \cdot I_{2n},$$
    where $G$ is a GOE random matrix and $J_{2n} = \big[\begin{smallmatrix}
        0 & 1 \\ -1 &0
    \end{smallmatrix}\big]^{\oplus n}$.
\end{defin*}
Note that one can choose freely the variance of the elements of the GOE ensemble, so in this way we obtain a one-parameter family of probability distributions on random quantum covariance matrices; we refer the reader to Definition \ref{def:RQCM-ensemble} for the precise statements. 

The rest of the paper is devoted to the study of this ensemble, with a focus on entanglement properties of random Gaussian states on $m+n$ modes. In particular, we show that:
\begin{itemize}
    \item the choice of shifting a given matrix to enforce the Heisenberg uncertainty relation is motivated by an optimization problem related to the operator norm (\cref{sec:closest-covariance})
    \item the distribution of the \emph{$m$-mode marginal} of a $(m+n)$-mode RQCM is part of the same ensemble of distributions (\cref{prop:marginal-RQCM} and \cref{prop:marginal-large-mn})
    \item the \emph{eigenvalues} of a RQCM converge, in the large number of modes limit, to a shifted semicircular distribution (\cref{thm:RQCM-eigenvalues})
    \item RQCMs (almost) satisfy the \emph{positive partial transposition} (PPT) criterion, in the same asymptotic regime (\cref{thm:PPT})
    \item the \emph{purity} of a RQCM decreases exponentially with the number of modes (\cref{prop:purity}) 
    \item the \emph{symplectic eigenvalues} of a RQCM converge to a probability distribution that can be characterized using free probability (\cref{thm:symplectic-eigenvalues})
    \item RQCM states have very interesting entanglement and extendability properties, that we probe numerically in \cref{sec:entanglement}, in different scenarios corresponding to several bipartitions of the number of modes and the variance parameter of the RQCM distribution. 
\end{itemize}

The results described above encapsulate most of the important properties of random Gaussian states. We argue that the ensemble of quantum covariance matrices that we introduced is natural, both from a physical perspective (it is invariant with respect to the action of the ortho-symplectic group) and from a mathematical perspective. Importantly, we relate the study of their properties to free probability theory \cite{voiculescu1992free,nica2006lectures, mingo2017free}.

Let us note that there has been earlier work on random Gaussian quantum states \cite{serafini,fukuda-koenig-typ}. Importantly, in the latter work, the authors take a completely different approach: they fix a set of symplectic eigenvalues, then they generate a random pure Gaussian state by randomly rotating the diagonal symplectic matrix, and finally they take the partial trace to obtain a general Gaussian state. In \cite{fukuda-koenig-typ}, the major hurdle is the fact that the group
\[\Sp(2n,\mathbb{R}) =\{ L \in M(2n, \mathbb{R}): L J_{2n}L^t =
J_{2n}\},\] 
which is the symmetry group for pure Gaussian states, is not compact; hence
sampling on this group is difficult. This approach is fundamentally different than ours: \emph{we construct directly the mixed quantum states, bypassing the pure states}, for which there is no canonical invariant probability distribution. Our goal is to study different properties of typical Gaussian states by using random matrix theory. The key part in this work is the novel sampling technique which will be discussed in \cref{sec:RQCM}. In Sections \ref{sec:large-RQCM}, \ref{sec:symplectic}, \ref{sec:entanglement} we study different spectral and entanglement properties of the newly introduced ensemble, using various techniques from \emph{random matrix theory} and \emph{free probability theory}. The results are presented in such a way that one can use them readily to perform analytical computations and numerical simulations for problems involving Gaussian states described by covariance matrices. Hence, we hope that our work provides a useful toolbox for researchers investigating entanglement properties of such quantum states, similarly to previous investigations in the discrete setting \cite{zyczkowski2011generating,MR4266112}.

\bigskip

The paper is organized as follows. In Section \ref{sec:background}, we are giving a short introduction of Gaussian states and entanglement, which will be used later in the paper. In Section \ref{sec:closest-covariance} we justify our approach to defining random quantum covariance matrices by showing that our method is equivalent to picking the closest (in operator norm) quantum covariance matrix to a GOE element. In the rest of \cref{sec:RQCM} we introduce random quantum covariance matrices and study their most basic properties. Sections \ref{sec:large-RQCM} and \ref{sec:symplectic} deal with properties related to usual and symplectic eigenvalues in the large number of modes asymptotic limit. Finally, Section \ref{sec:entanglement}, more numerical in nature, gathers results about the entanglement of typical quantum covariance matrices.

A \textsf{Mathematica} notebook containing numerical routines to sample random quantum covariance matrices and to test important entanglement-related properties is available at \cite{notebook}.

\section{Gaussian states and their entanglement}\label{sec:background}

A quantum state is a positive semidefinite trace class operator with
trace 1, in a separable complex Hilbert space. Consider the `$n$-mode'
Fock space $\Gamma(\mathbb{C}^n)$. A Gaussian state $\rho$ acting on
this space is a state such that its expectation with respect to the
position or momentum operator in any mode gives a classical Gaussian
distribution. In this work, we plan to explore properties of Gaussian
states. 

A state $\rho$ in $\Gamma(\mathcal{H})$ with a Hilbert space $\mathcal{H}=
\mathbb{C}^n$ is an $n$-mode Gaussian state if its
Fourier transform $\hat{\rho}$ is given by 
\begin{equation}\label{eq:1.4}
\hat{\rho}(\bm{x} + \imath \bm{y}) = \exp
\left[-\imath \sqrt{2} (\bm{l}^T\bm{x} -
\bm{m}^T\bm{y}) -\frac{1}{2} \left[\tau\begin{pmatrix} \bm{x}\\
\bm{y} \end{pmatrix}\right]^T S  \left[\tau\begin{pmatrix} \bm{x}\\
\bm{y} \end{pmatrix}\right] \right].
\end{equation}
for all $\bm{x},~\bm{y}\in \mathbb{R}^n$ where $\bm{l},~
\bm{m}$ are the momentum-position mean vectors and $S/2$ their
covariance matrix. $\tau:\mathbb{R}^{2n} \rightarrow
\mathbb{R}^{2n}$ is given by $\tau(x_1, \cdots, x_n, y_1,
\cdots, y_n)^T= (x_1,y_1, \cdots, x_n, y_n)^T$ \cite{krpg}.

\par If $\rho$ is a state of a quantum system and
$X_i,\,i=1,2$ are two real-valued observables, or
equivalently, self-adjoint operators with finite second
moments in the state $\rho$ then the covariance between $X_1$
and $X_2$ in the state $\rho$ is the scalar quantity
\[\tr\left(\frac{1}{2}(X_1 X_2 +X_2 X_1)\rho\right) - \left(\tr X_1
\rho\right) \cdot \left(\tr X_2\rho\right),\]
 which is denoted by
$\mathrm{Cov}_\rho (X_1,X_2)$. Suppose $q_1,p_1;\,
q_2,p_2;\, \cdots ;\, q_n,p_n$ are the position - momentum
pairs of observables of a quantum system with $n$ degrees of
freedom obeying the canonical commutation relations. Then
we express 
\[(X_1,X_2,\cdots,X_{2n}) =
(q_1,-p_1,q_2,-p_2,\cdots,q_n,-p_n).\]
If $\rho$ is a state in which all the $X_j$'s have finite
second moments we write 
\begin{equation}\label{eq:1.1}
S_\rho = \Big[ \mathrm{Cov}_\rho(X_i, X_j)\Big]_{i,j=1}^{2n}.
\end{equation} 
We call $S_\rho$ the covariance matrix of the position
momentum observables. If we write 
\begin{equation}\label{eq:1.2}
J_{2n}= \begin{bmatrix}\begin{array}{rr}
0 & 1 \\ -1 & 0
\end{array} &&& \\
& \begin{array}{rr}
0 & 1 \\ -1 & 0
\end{array} &&\\ 
 &  & \ddots & \\
&&& \begin{array}{rr}
0 & 1 \\ -1 & 0
\end{array}
\end{bmatrix} = \begin{bmatrix} 0 & 1 \\ -1 & 0 \end{bmatrix}\otimes I_n
\end{equation}
where $I_n$ is the identity matrix of order $n$; or equivalently we may write it as $\big[\begin{smallmatrix}
        0 & 1 \\ -1 &0
    \end{smallmatrix}\big]^{\oplus n}$ for the $2n \times 2n$ block diagonal
matrix, the complete Heisenberg uncertainty relations for
all the position and momentum observables assume the form
of the following matrix inequality 
\begin{equation}\label{eq:1.3}
S_\rho + \ii J_{2n} \geq 0.
\end{equation} 
Conversely, if $S$ is any real $2n \times 2n$ symmetric
matrix obeying the inequality $S +  \ii
J_{2n} \geq 0$, then there exists a state $\rho$ such that
$S$ is the covariance matrix $S_\rho$ of the observables
$q_1,-p_1;\,q_2,-p_2;\, \cdots ;\, q_n,-p_n$. In such a case $\rho$ can be
chosen to be a Gaussian state with mean zero. For a $n$-mode covariance matrix $S$ the positive eigenvalues of the matrix $\ii J_{2n} S$ are called the \emph{symplectic eigenvalues} of $S$.

\par The importance of finite mode Gaussian states and their
covariance matrices in general quantum theory as well as
quantum information has been highlighted extensively in the
literature. A comprehensive survey of Gaussian states and
their properties can be found in the
book of Holevo \cite{MR2797301}. For their applications to 
quantum information theory the reader is referred to the
survey article by Weedbrook et al \cite{RevModPhys.84.621}
as well as Holevo's book \cite{ MR2986302}.  For our
reference we use \cite{arvindsurvey, krpg,
MR3070484} for Gaussian states. For notations in the
following sections we use \cite{krprb} and \cite{krprb2}.

\subsection{Entanglement of Gaussian states}

\par One of the most important problems in quantum mechanics
as well as quantum information theory is to determine whether a
given bipartite state is separable or entangled \cite{NC}.
There are several methods in tackling this problem leading
to a long list of important publications. A detailed
discussion on this topic is available in the survey articles
by Horodecki et al \cite{RevModPhys.81.865}, and 
G{\"u}hne and T{\'o}th \cite{guth}. One such condition which
is both necessary and sufficient for separability in finite
dimensional product spaces is complete extendability
\cite{PhysRevA.69.022308}. Let us denote by $\mathcal{B(H)}$ the set of bounded operators on a Hilber space $\mathcal{H}$.

\begin{defin}\label{def:1}
Let $k \in \mathbb{N}$. A state $\rho \in
\mathcal{B}(\mathcal{H}_A \otimes \mathcal{H}_B)$ is said to
be $k$-extendable with respect to system $B$ if there is a
state $\tilde{\rho} \in \mathcal{B}(\mathcal{H}_A \otimes
\mathcal{H}_B^{\otimes k})$ which is invariant under any
permutation in $\mathcal{H}_B^{\otimes k}$ and $\rho =
\tr_{\mathcal{H}_B^{\otimes (k-1)}}\tilde{\rho}$, $k \ge 2$. 
\par A state $\rho \in \mathcal{B}(\mathcal{H}_A \otimes
\mathcal{H}_B)$ is said to be completely extendable if it
is $k$-extendable for all $k \in \mathbb{N}$. 
\end{defin}
The following theorem of Doherty, Parrilo, and
Spedalieri \cite{PhysRevA.69.022308} emphasizes the
importance of the notion of complete extendability.

\begin{customthm}{A}\cite{PhysRevA.69.022308} \label{th:a}
A bipartite state $\rho \in \mathcal{B}(\mathcal{H}_A \otimes
\mathcal{H}_B)$ is separable if and only if it is completely
extendable with respect to one of its subsystems.
\end{customthm}

\par It is fairly simple to see that separability implies complete
extendability. The proof of the converse depends
on an application of the quantum de Finetti theorem
\cite{MR0241992, MR0397421}, according to which any
exchangable state is, indeed, separable. The link between separability and
extendability has found applications in quantum information theory
\cite{MR2825510, MR3210848}. Here we study the same in
the context of quantum Gaussian states.

This proof directly follows from work of Hudson and Moody
\cite{MR0397421}. The problem of finding necessary and
sufficient conditions for $k$-extendability of non-Gaussian states is
open. 

\begin{defin}[Gaussian extendability]\label{def:2}
Let $k \in \mathbb{N}$. A Gaussian state $\rho_g$ in
$\Gamma(\mathbb{C}^m) \otimes \Gamma(\mathbb{C}^n)$ is said to
be Gaussian $k$-extendable with respect to the second system  if there is a
Gaussian state $\tilde{\rho_g}$ in  $\Gamma(\mathbb{C}^m) \otimes
\Gamma(\mathbb{C}^n)^{\otimes k}$  which is invariant under any
permutation in $\Gamma(\mathbb{C}^n)^{\otimes k}$ and $\rho_g =
\tr_{\Gamma(\mathbb{C}^n)^{\otimes (k-1)}}\tilde{\rho_g}$,
$k\ge 2$. 
\par A Gaussian state $\rho_g$ in $\Gamma(\mathbb{C}^m)
\otimes \Gamma(\mathbb{C}^n)$ is said to be Gaussian
completely extendable if it is Gaussian $k$-extendable
for every $k \in \mathbb{N}$.
\end{defin}
\par Entanglement property of a Gaussian state depends only on
its covariance matrix. Hence without loss of generality, we
can confine our attention to the Gaussian states with
mean zero. Thus an $(m+n)$-mode mean zero Gaussian state in
$\Gamma(\mathbb{C}^m) \otimes \Gamma(\mathbb{C}^n)$ is
uniquely determined by a $2(m+n) \times 2(m+n)$ covariance
matrix 
\begin{equation}\label{eq:m+n-modes}
    S = \begin{bmatrix} A & B \\ B^T & C \end{bmatrix}.
\end{equation}
In this paper, we shall call the matrices above \emph{quantum covariance matrices}, following the terminology from \cite{lami2018gaussian}. Here $A$ and $C$ are covariance matrices of the $m$ and $n$-mode
marginal states respectively. 

\par If $\rho(\bm{0},\bm{0};S)$, written in short as
$\rho(S)$ in $\Gamma(\mathbb{C}^m) \otimes
\Gamma(\mathbb{C}^n)$ is $k$-extendable with respect to the
second system, then there exists a real matrix $\theta_k$ of
order $2n \times 2n$ such that the extended matrix
\begin{equation}\label{eq:1.5}
S_k = \left[ \begin{array}{c|cccc}
A & B & B & \cdots & B \\\hline
B^T & C & \theta_k & \cdots & \theta_k \\
B^T & \theta_k^T & C & \cdots & \theta_k\\
\vdots & \vdots & \vdots & \ddots & \vdots\\
B^T & \theta_k^T & \theta_k^T & \cdots & C
\end{array} \right]
\end{equation} 
is the covariance matrix of a Gaussian state in
$\Gamma(\mathbb{C}^m) \otimes \Gamma(\mathbb{C}^n)^{\otimes
k}$. Then it satisfies inequality (\ref{eq:1.3}) in the form
\begin{equation}\label{eq:1.6}
S_k + \ii J_{2(m+kn)} \ge 0.
\end{equation}

\begin{customthm}{B}\label{th:1}\cite{krprb3}
Let $\rho$ be a bipartite Gaussian state in $\Gamma(\mathbb{C}^m)
\otimes \Gamma(\mathbb{C}^n)$ with covariance matrix $S =
\begin{bmatrix} A & B \\ B^T & C \end{bmatrix}$, where $A$
and $C$ are marginal covariance matrices of the first and second
system respectively. Then $\rho$ is completely extendable with respect
to the second system if and only if there exists a real positive
matrix $\theta$ such that 
\begin{equation}\label{eq:xx} 
C+ \ii J_{2n} \ge \theta \ge B^T \left( A + \ii  J_{2m}
\right)^{-} B,
\end{equation} 
where $\left( A + \ii J_{2m}\right)^{-}$ is the
Moore-Penrose inverse of $A + \ii  J_{2m}$.
\end{customthm}
\noindent From this the following result can be constructed. 
\begin{customthm}{C}\label{th:2}
Any separable Gaussian state in a bipartite system is
completely extendable and conversely every completely extendable
Gaussian state is separable. 
\end{customthm}
The authors of the paper \cite{krprb3} have shown that the result in
 \cref{th:2} is true for any state (need not be Gaussian) in
bipartite Fock space as well. However it may not be possible to get
any analogous matrix inequality as entanglement in such systems is
more completed than the finite dimensional versions. Later Lami et al
\cite{lami-prl-19} gave an improved version of the \cref{th:1}
where they gave explicit bounds for the different level of
entanglement. The result is as follows. 
\begin{customthm}{D}\label{th:3} 
Let $\rho_{AB}$ be a $k$-extendible (not necessarily
Gaussian) state of $m + n$ modes with covariance matrix $S$.
Then there exists a $2n \times 2n$ quantum covariance matrix
$\Delta \ge \ii J_{2n}$ for the space $B$ such that
\begin{equation}\label{xx2}
S \ge \ii J_{2m} \oplus \left[\left(1 - \frac{1}{k}\right) \Delta
+ \frac{1}{k}\ii J_{2n}\right].
\end{equation}
Moreover, the above condition is necessary and sufficient for
$k$-extendibility when $\rho_{AB}$ is Gaussian. In this case the
equation \eqref{xx2} can further be simplified as 
\begin{equation}\label{xx3}
\ii J_{2n} \le \Delta \le \frac{k}{k-1}\left[C - B^T ( A - \ii
J_{2m})^{-1} B\right] - \frac{1}{k-1}  \ii J_{2n}. 
\end{equation}
Furthermore, if $\rho_{AB}$ is an $m+n$ mode $k$-extendible Gaussian state, then 
\begin{equation}\label{eq:k-extendible_distance}
\|\rho_{AB}-\mathrm{SEP}(m,l)\|_1 \leq 2n/k,
\end{equation}
where $\mathrm{SEP}(m,n)$ is the set of bipartite separable states of the systems of $m$ and $n$ modes and $||\cdot||_1$ is the trace norm.
\end{customthm}
 
\par Both the complete extendability and $k$-compatibility criteria can be cast as semidefinite programs (SDPs). While one can in principle determine the separability of any covariance matrix by running the SDPs, for sufficiently large matrices solving the SDPs might become cumbersome because of the excessive run-time. However, there are also other, perhaps simpler, entanglement criteria. One of the most well-known is the \emph{positive partial transpose (PPT) criteria}. The  PPT criteria of Gaussian states can be expressed in the following
way. Consider a $(m+n)$-mode Gaussian state with covariance matrix $S$
as given in the Theorem \ref{th:2}. It naturally satisfies the equation
\eqref{eq:1.2} with appropriate dimension. It is considered a PPT state if 
\begin{equation}\label{eq:PPT-condition} 
S+ \ii
\begin{bmatrix}
J_{2m} & \\
 &-J_{2n}
\end{bmatrix} \ge 0.
\end{equation}
Simon \cite{PhysRevLett.84.2726} showed that a 2-mode Gaussian state
is separable if and only if its covariance matrix $S$ satisfies
\eqref{eq:PPT-condition}. Furthermore it has been proven that this also holds
for any $(n+1)$-mode Gaussian state where the bipartition is taken as
in the $n$-mode vs 1-mode way. However, in general the PPT condition is not necessary and
sufficient for separability. In fact, Werner and Wolf \cite{PhysRevLett.86.3658}
constructed  examples of Gaussian states on 2-mode $\times$ 2-modes
settings which satisfies \eqref{eq:PPT-condition} but is entangled.

\subsection{Previous work on random Gaussian states}  
Here we give the two different
methods of sampling given in Fukuda and K\"onig \cite{fukuda-koenig-typ} using techniques proposed earlier by Serafini et al
\cite{MR2345312}. 

Consider the system $\mathcal{H}_A \otimes \mathcal{H}_B =
L^2(\mathbb{R})^{\otimes k} \otimes L^2(\mathbb{R})^{\otimes (n-k)}$,
where the system is of $k$-modes with $(n-k)$-modes of environment.
Consider a pure Gaussian state $\ket{\Phi} \in \mathcal{H}_A \otimes \mathcal{H}_B$
bounded by the compactness criteria (i.e. energy constrain) 
\[\langle \Phi |H_{AB}| \Phi \rangle \le E,\]
with the Hamiltonian 
\[H_{AB} = \sum_{j=1}^n (Q_j^2 + P_j^2).\]
This can be diagonalised by \emph{passive} way, i.e. by using the
ortho-symplectic group $\K(2n) := \Sp(2n) \cap
\O(2n)$\footnote{Note that $K(2n)$ consists of the operators
in $\Sp(2n)$ which are orthogonal. This is the largest compact
subgroup of $\Sp(2n)$. Hence an unique  Haar measure can be defined on
it which can be used for sampling.}  we can
write this as the diagonal matrix $\hat{Z_n} = Z_n  \oplus Z_n^{-1}$
where $Z_n = \operatorname{diag}(z_1, \cdots, z_n)$, with $z_j \ge 1$ for all $j$, and   
\[\sum_{j=1}^n \left( z_j +\frac{1}{z_j} \right) \le E.\] 
The left hand side is actually the energy $\langle
\Phi|H_{AB}|\Phi\rangle $ of the state $\ket{\Phi}$. The method can be
expressed as fellows. First to choose $z_j$'s randomly with required
properties and bounds, which in turn gives a pure state in the state
space. After this we may apply a random ortho-symplectic
transformation. This is allowed as the group $\K(2n)$ is compact.
Following this we may remove the extra $(n-k)$ modes associated with
environment by using partial trace to generate a $k$-mode random mixed
Gaussian state. 

\bigskip

Previously, Serafini et al \cite{MR2345312} considered the two following
measured based on the above protocol. 
\begin{description}
\item[Microcanonical measure $\mu_{micro}$] This is done by first drawing $(E_1 ,
\cdots, E_n )$ uniformly (according to the measure induced by the Lebesgue measure on $\mathbb{R}^n$) from the set 
\[\Gamma_E = \left\lbrace(E_1, \cdots, E_n): E_j \ge 2 ~\forall j, ~
\sum_{j=1}^n E_j \le E\right\rbrace.\]
Then 
\[z_j = \frac{1}{2} (E_j + \sqrt{E_j^2 -4}),\]
for $j=1, \cdots, n$ and drawing pure states from $\mu_{micro}$. 
\item[Canonical measure $\mu_{canonical}$] This is achieved by drawing
$E=(E_1, \cdots, E_n)$ based on Boltzmann distribution 
\[\dd p(E)= \frac{1}{T^n} \exp\left[-\frac{1}{T} \left(\sum_{j=1}^n E_j -2n
\right)\right]\dd E_1 \cdots \dd E_n\]
with temperature $T= \frac{E}{n}$ and setting 
\[z_j = \frac{1}{2} (E_j + \sqrt{E_j -4})\]
for $j=1, \cdots, n$ and drawing pure states. 
\end{description}

\section{Random quantum covariance matrices}\label{sec:RQCM}

In this section, we shall introduce a random matrix ensemble of Gaussian covariance matrices. Our inspiration comes from the discrete case, where such ensembles of random density matrices have found many applications in quantum information theory and beyond. Ensembles of (finite dimensional) density matrices \cite{braunstein1996geometry,hall1998random,zyczkowski2001induced,sommers2004statistical,zyczkowski2011generating} have been studied thoroughly in relation to many topics such as quantum entanglement or quantum chaos. Ensembles of quantum channels have been successfully used to prove many important results in quantum information theory \cite{hayden2006aspects,hastings2009superadditivity}; see \cite{collins2016random} for a review of these topics.

\medskip

Recall that a Gaussian covariance matrix is a \emph{real} symmetric matrix $S \in \mathcal M_{2n}(\mathbb R)$ satisfying the condition
$$S \geq \ii J_{2n}.$$
Mathematically, the condition above is interesting since it requires a \emph{complex} positive semidefiniteness condition of a real object. A random quantum covariance matrix will be defined to be an element of the Gaussian Orthogonal Ensemble (GOE) shifted by a multiple of the identity matrix in order for the condition $S \geq \ii J_{2n}$ to be satisfied. Our approach is motivated by the following two independent facts: 
\begin{itemize}
    \item the ensemble of random quantum covariance matrices we construct has is invariant under the ortho-symplectic group $\K(2n) = \Sp(2n) \cap \O(2n)$, see \cref{prop:ortho-symplectic-invariance};
    \item to a GOE element $G$, we associate the closest (in operator norm) quantum covariance matrix $H$, see \cref{sec:closest-covariance}.
\end{itemize}

We start by investigating the problem of finding the closes quantum covariance matrix to a given matrix. After recalling some basic facts about GOE random matrices, we introduce the ensemble of random quantum covariance matrices (RQCM) along with some of its basic properties. Later, in Section \ref{sec:large-RQCM} we study the large $n$ limit of this ensemble.

\subsection{Closest quantum covariance matrix}\label{sec:closest-covariance}

Consider an arbitrary real symmetric matrix $G \in \mathcal M_{2n}^{sa}(\mathbb R)$. We are interested in finding the closest \emph{quantum covariance matrix} to $G$, that is
\begin{align}
\nonumber\min &\quad\|G-H\| \\
\label{eq:closest-covariance}\text{s.t.} &\quad H \in \mathcal M_{2n}^{sa}(\mathbb R) \\
\nonumber &\quad H \geq \ii J_{2n},
\end{align}
where $\mathcal{M}^{sa}_{2n}(\mathbb{R})$ denotes the set of $2n \times 2n$ real and self-adjoint, i.e., symmetric matrices.

The geometry of the optimization problem above depends on the norm used as a cost function. We consider here the operator norm $\| \cdot \|_\infty$. The optimization problem \eqref{eq:closest-covariance} becomes an SDP: 
\begin{align*}
\min & \quad t \\
\text{s.t.} &\quad H \in \mathcal M_{2n}^{sa}(\mathbb R), t \in \mathbb R \\
&\quad H \geq \ii J_{2n}\\
&\quad -t I_{2n} \leq G-H \leq t I_{2n}.
\end{align*}

The dual SDP can be easily computed: 
\begin{align}
\nonumber\max & \quad \langle Z, iJ_{2n}-G \rangle \\
\label{eq:SDP-infty-dual}\text{s.t.} &\quad Z \in \mathcal M_{2n}^{sa}(\mathbb C)\\
\nonumber&\quad Z \geq 0\\
\nonumber&\quad \|\operatorname{Re} Z \|_1 \leq 1.
\end{align}

However, since $Z \geq 0$, we also have $\operatorname{Re} Z \geq 0$ and thus the last condition above reads simply $\operatorname{Tr} Z \leq 1$, yielding the explicit solution 
\begin{equation}\label{eq:solutiuon-SDP}
    \lambda(G) = \max(0, \lambda_{\max}(\ii J_{2n}-G)).
\end{equation}
In particular, we see that in the case when $G \geq \ii J_{2n}$, one can simply take $H=G$, obtaining the optimum 0. 

Importantly, there is a trivial solution $H = G + \lambda(G) I_{2n}$ achieving the value above. Indeed, write the decomposition of the matrix $G-\ii J_{2n}$ into its positive and negative parts:
$$G- \ii J_{2n} = A-B,$$
with $A, B \geq0$ having orthogonal supports. Clearly, $\lambda(G) = \lambda_{\max}(B) \geq 0$, and we have 
$$G + \lambda_{\max}(B)I_{2n} = \ii J_{2n} + A  +(\lambda_{max}(B)I_{2n}-B) \geq \ii J_{2n},$$
proving that $G + \lambda_{\max}(B)I_{2n}$ achieves the optimum value in the primal program. The geomtry of the problem and the closest matrix $H$ described above are presented graphically in Figure \ref{fig:closest-QCM}. 

\begin{figure}
    \centering
    \includegraphics{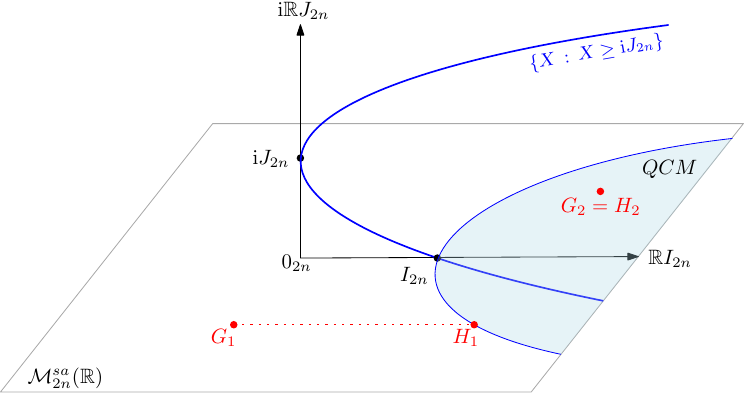}
    \caption{Finding the closest (in operator norm) quantum covariance matrix $H$ to a given real symmetric matrix $G$. We display two examples: $G_1$ is not a quantum covariance matrix, so it needs to be shifted by a multiple of the identity matrix to become one ($H_1$); $G_2$ is already a quantum covariance matrix, so $H_2=G_2$.}
    \label{fig:closest-QCM}
\end{figure}

\subsection{Basic properties of GOE matrices}

The \emph{Gaussian orthogonal ensemble} $\GOE$ is arguably the most studied ensemble of random matrices. It consists of real symmetric matrices distributed along the Gaussian distribution on the corresponding vector space. 

\begin{defin}
    A random symmetric matrix $G \in \mathcal M^{sa}_n(\mathbb R)$ is said to have a $\GOE(n, \sigma)$ distribution if:
    \begin{itemize}
        \item the random variables $\{G_{ij}\}_{1 \leq i \leq j \leq n}$ are independent;
        \item the diagonal entries have distribution
        $$\forall i \in [n], \qquad G_{ii} \sim \mathcal N(0, 2\sigma^2)$$
        \item the off-diagonal entries have distribution
        $$\forall i<j \in [n], \qquad G_{ij} = G_{ji} \sim \mathcal N(0, \sigma^2).$$
    \end{itemize}
\end{defin}

On the space of real symmetric matrices, random $\GOE(n, \sigma=1)$ matrices have distribution given by \cite[Section 2.5.1]{anderson2010introduction}
$$\frac{\mathrm{d}G}{\mathrm{d} \mathrm{Leb}} = 2^{-n/2} (2\pi)^{-n(n+1)/4} \exp\big( - \operatorname{Tr}(G^2)/4 \big),$$
where $\mathrm{Leb}$ is the Lebesgue measure on the vector space of $n \times n$ real symmetric matrices. 

Wigner famously showed that the empirical eigenvalue distribution 
$$\mu_G:= \frac 1 n \sum_{i=1}^n \delta_{\lambda_i(G)}$$
converges to the (centered) \emph{semicircle distribution} defined by
$$\mathrm{d} \mathrm{SC}_\sigma = \frac{\sqrt{4\sigma^2-x^2}}{2 \pi \sigma^2}  \mathbf{1}_{|x| \leq 2\sigma}(x) \mathrm{d}x.$$

\begin{theorem}{\cite{wigner1955characteristic}}\label{thm:GOE-limit}
    Let $\sigma >0$ and $G_n \sim \GOE(n, \sigma)$ be a sequence of $\GOE$ random matrices. Then, almost surely,
    $$\lim_{n \to \infty} \mu_{G_n / \sqrt n} = \mathrm{SC}_\sigma,$$
    where the convergence of probability measures is considered in the weak sense. 
\end{theorem}

We plot in Figure \ref{fig:GOE-histogram} the empirical histogram of eigenvalues versus the theoretical curve of the semicircular distribution. Note also that the largest and the smallest eigenvalues of a $\GOE$ matrix also converge to the edges of the support of the semicircular distribution. 

\begin{figure}
    \centering
    \includegraphics[width=.47\textwidth]{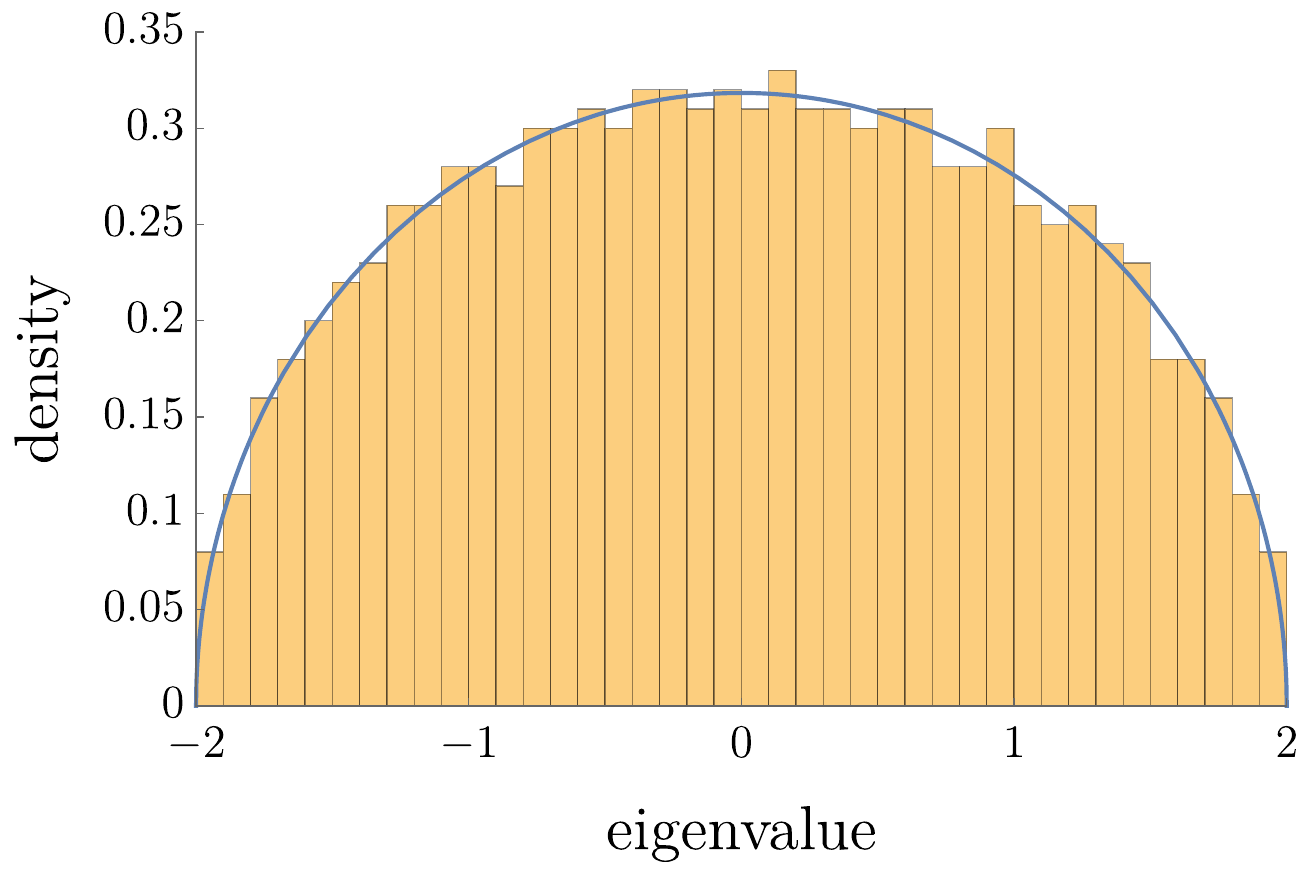} \qquad 
    \includegraphics[width=.47\textwidth]{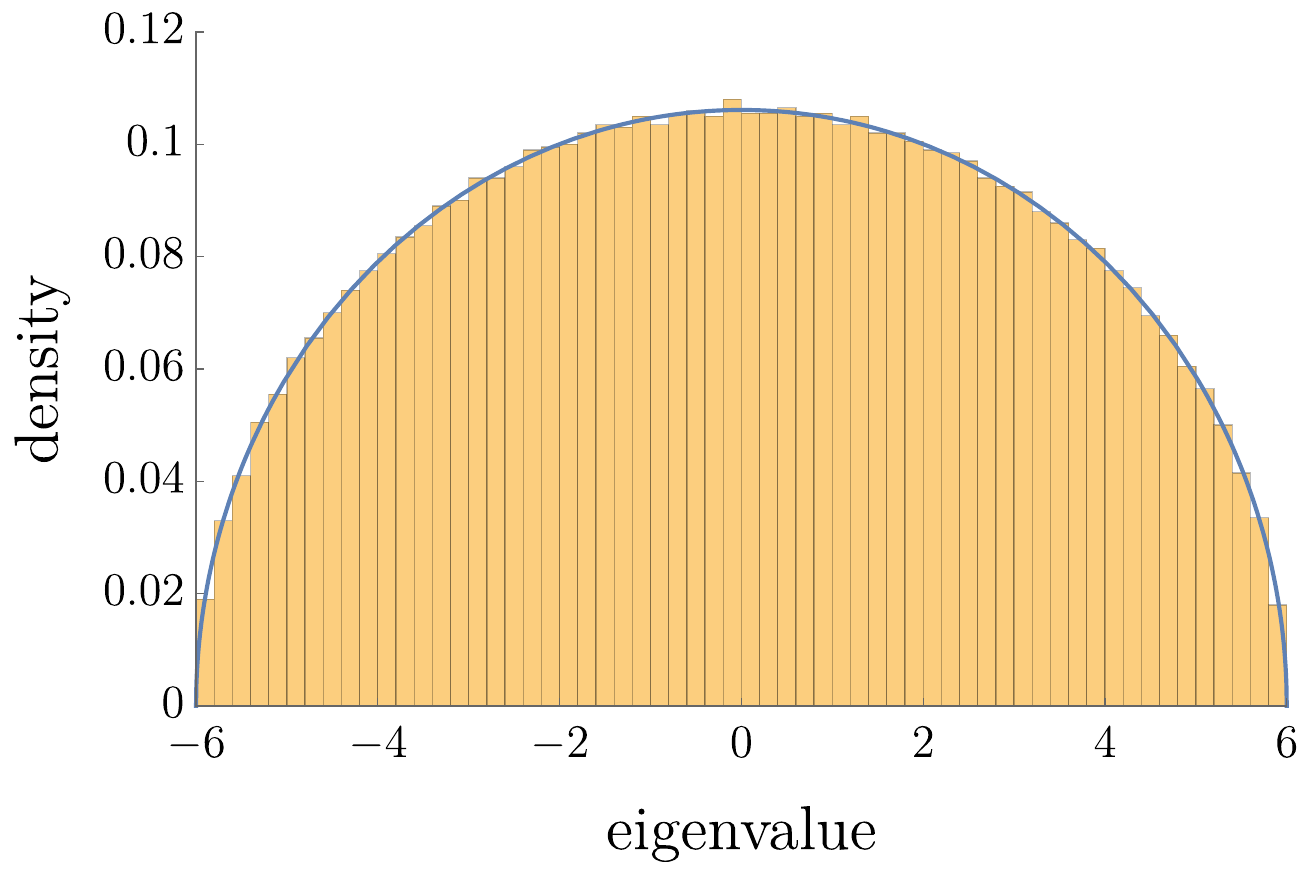}
    \caption{Theoretical curve (blue) and histogram of eigenvalues of one realization of $\GOE$ matrices (yellow). Left: $n=500$, $\sigma = 1$; right: $n=5000$, $\sigma = 3$.}
    \label{fig:GOE-histogram}
\end{figure}

\begin{theorem}{\cite{furedi1981eigenvalues}}
    Let $\sigma >0$ and $G_n \sim \GOE(n, \sigma)$ be a sequence of $\GOE$ random matrices. Then, almost surely,
    $$\lim_{n \to \infty} \lambda_{\min}(G_n / \sqrt n) = -2 \sigma \quad \text{ and } \quad \lim_{n \to \infty} \lambda_{\max}(G_n / \sqrt n) = 2 \sigma.$$
\end{theorem}

\subsection{Random Quantum Covariance Matrices --- RQCM}

As motivated at the beginning of the section, we shall define the ensemble of random quantum covariance matrices associating to an element $G$ of the GOE ensemble the closest quantum covariance matrix, see \cref{eq:closest-covariance}. Note however that if $G$ is already a quantum covariance matrix, we still shift $G$ by a multiple of the identity. 

\begin{defin}\label{def:RQCM-ensemble}
    For a given $n \geq 1$, let $G \in \GOE(2n, \sigma)$ be an element of the Gaussian Orthogonal Ensemble with mean $0$ and variance $\sigma^2$. A \emph{random quantum covariance matrix} of parameter $\sigma$ is a matrix
    $$S_G:=G + \lambda_{\max}(\ii J_{2n} - G) \cdot I_{2n}.$$
    We denote by $\RQCM(2n, \sigma)$ the ensemble of such random matrices.
\end{defin}

Clearly, every element of $\RQCM(2n, \sigma)$ is a Gaussian covariance matrix: 
$$\lambda_{\min}(S_G-\ii J_{2n}) = \lambda_{\min}(G -\ii J_{2n}) + \lambda_{\max}(\ii J_{2n} - G) = 0. $$

The $\RQCM$ ensemble has $\K(2n) = \O(2n) \cap \Sp(2n)$ symmetry, as it is shown in the following proposition. 

\begin{prop}\label{prop:ortho-symplectic-invariance}
For any real matrix $G \in \mathcal M_{2n}(\mathbb R)$ and any operator $U \in \K(2n)$, we have 
$$S_{UGU^\top} = U S_G U^\top.$$
In particular, if $S$ is a random matrix having $\RQCM(2n, \sigma)$ distribution, then so does $UGU^\top$. In other words, the ensemble $\RQCM(2n, \sigma)$ is invariant with respect to the ortho-symplectic group.
\end{prop}
\begin{proof}
    Compute
    \begin{align*}
    S_{UGU^\top} &= UGU^\top + \lambda_{\max}(\ii J_{2n} - UGU^\top)I_{2n}\\
    &= UGU^\top + \lambda_{\max}\bigg(U(\ii \underbrace{U^\top J_{2n}U}_{J_{2n}} - G)U^\top\bigg)I_{2n}\\
    &= UGU^\top + \lambda_{\max}(\ii J_{2n} - G)I_{2n}\\
    &= U\bigg(G+ \lambda_{\max}(\ii J_{2n} - G)I_{2n} \bigg)U^\top \\
    &= U S_G U^\top.
    \end{align*}

    For the second claim, if $S$ has $\RQCM(2n, \sigma)$ distribution, then $S=S_G$ for a random matrix $G$ having $\GOE(2n, \sigma)$ distribution. It follows that $UGU^\top$ also has $\GOE(2n, \sigma)$ distribution, and thus 
    $USU^\top = U S_G U^\top = S_{UGU^\top}$ has $\RQCM(2n, \sigma)$ distribution.
\end{proof}

Let us consider now the partial trace operator, where, given a Gaussian covariance matrix $S \in \mathcal M_{2n}(\mathbb R)$ corresponding to a $n$-mode Gaussian state, we associate its $m$-mode marginal $S^{(2m)}$, obtained as its top-left $2m \times 2m$ block, see Eq.~\eqref{eq:m+n-modes}. The following result shows that the $\RQCM$ ensemble is, up to translations, stable under taking marginals. 

\begin{prop}\label{prop:marginal-RQCM}
    Let $S$ be a random quantum covariance matrix having $\RQCM(2n, \sigma)$ distribution, and let $1 \leq m <n$. Then, its $m$-mode marginal has a \emph{shifted} $\RQCM(2m, \sigma)$ distribution.
\end{prop}
\begin{proof}
    Writing $S = S_G$ for $G \sim \GOE(2n, \sigma)$, we have 
    $$(S_G)^{(2m)} = G^{(2m)} + \lambda_{\max}(\ii J_{2n} - G)I_{2m} = S_{G^{(2m)}} + \delta I_{2m},$$
    where 
    $$\delta := \lambda_{\max}(\ii J_{2n} - G) - \lambda_{\max}\big( (\ii J_{2n} - G)^{(2m)} \big) \geq 0.$$
    Since the $2m \times 2m$ top-left corner of a $\GOE(2n, \sigma)$ is a $\GOE(2m, \sigma)$ matrix, the claim follows. 
\end{proof}

\section{The large number of modes limit of random quantum covariance matrices}\label{sec:large-RQCM}

We discuss in this section behavior of the $\RQCM$ ensemble in the large dimension limit, i.e.~in the limit where the number of modes $n$ goes to infinity. 

The first question we shall address is the large $n$ behavior of the quantity 
$$\lambda_{\max}(\ii J_{2n}-G)$$
from Definition \ref{def:RQCM-ensemble}. This quantity is the \emph{minimal} amount by which one needs to shift a $\GOE$ element in order to obtain a quantum covariance matrix. In other words, this is the minimal shift such that the smallest \emph{symplectic} eigenvalue of the resulting matrix is 1. 

\begin{prop}\label{prop:convergence-shifted-GOE}
Let $G_{2n} \in \GOE(2n, \sigma)$ an element from the GOE ensemble having variance $\sigma$. Then, almost surely, the random matrix 
$$\ii J_{2n} - \frac{G_{2n}}{\sqrt{2n}}$$
converges strongly, as $n \to \infty$ to the measure $\mu_\sigma$ given by
\begin{equation}
    \mu_\sigma := \Big( \frac 1 2 \delta_{-1} + \frac 1 2 \delta_1 \Big) \boxplus \mathrm{SC}_{\sigma},
\end{equation}
 where $\boxplus$ denotes Voiculescu's \emph{free additive convolution}. The measure $\mu_\sigma$ has support
 $$\operatorname{supp}(\mu_\sigma) = \begin{cases}
        [-R(\sigma), -L(\sigma)] \sqcup  [L(\sigma), R(\sigma)]&\qquad \text{ if } \sigma < 1\\
        [-R(\sigma), R(\sigma)] &\qquad \text{ if } \sigma \geq 1,
 \end{cases}$$
 with 
 \begin{align}
\label{eq:def-R}    R(\sigma) &:= \Big(1 + \frac \sigma 4 (\sqrt{8 + \sigma^2}-\sigma) \Big)\sqrt{1 + \frac \sigma 2 (\sqrt{8 + \sigma^2}+\sigma)} \qquad \forall \, \sigma > 0\\
\label{eq:def-L}    L(\sigma) &:=\Big(1 - \frac \sigma 4 (\sqrt{8 + \sigma^2}+\sigma) \Big)\sqrt{1 - \frac \sigma 2 (\sqrt{8 + \sigma^2}-\sigma)}\qquad \forall \, \sigma \in (0,1). 
 \end{align}
In particular, we have, almost surely,
$$\lim_{n \to \infty} \lambda_{\max}\Big(\ii J_{2n} - \frac{G_{2n}}{\sqrt{2n}}\Big) = R(\sigma).$$
\end{prop}
\begin{proof}
    We shall first compute the asymptotic limit of the empirical eigenvalue distribution of the random matrix 
    $$X_n := \ii J_{2n} - \frac{G_{2n}}{\sqrt{2n}}.$$
    From Theorem \ref{thm:GOE-limit}, we know that the asymptotic limit of the second term $G_n/\sqrt{2n}$ is the semicircular distribution $\mathrm{SC}_\sigma$. The first term, $\ii J_{2n}$ converges to the \emph{Bernoulli distribution}
    \begin{equation}\label{eq:Bernoulli}
        \ii J_{2n} \to \frac 1 2 \delta_{-1} + \frac 1 2 \delta_1.
    \end{equation}
    Hence, by Voiculescu's theorem about sums of unitarily invariant random matrices (see, e.g., \cite[Theorem 22.35]{nica2006lectures}), 
    $$X_n \to \mu_\sigma = \Big( \frac 1 2 \delta_{-1} + \frac 1 2 \delta_1 \Big) \boxplus \mathrm{SC}_{\sigma},$$
  where we have used implicitly the symmetry of the semicircular distribution. The study of the free additive convolutions of semicircular distributions was initiated by Biane \cite{biane1997free} using the \emph{subordination property}. Here, we use the (equivalent) statements from \cite[Theorem 2.1]{capitaine2011free}, with $\nu$ being the Bernoulli distribution from Eq.~\eqref{eq:Bernoulli}: the complementary support of the measure $( \frac 1 2 \delta_{-1} + \frac 1 2 \delta_1 ) \boxplus \mathrm{SC}_{\sigma}$ is given by
    $$H_\sigma(\mathbb R \setminus U_\sigma),$$
    where $H_\sigma$ is the function 
    $$H_\sigma(z) := z + \frac{\sigma^2 z}{z^2-1}$$
    and $U_\sigma$ is the set
    $$U_\sigma:=\Big\{ u \in \mathbb R \, : \, \frac{u^2+1}{(u-1)^2(u+1)^2} > \frac{1}{\sigma^2}\Big\}.$$
    Solving for $u$ in the equation above and plugging the values in the formula for $H_\sigma$ yield the values for $L,R$ from the statement. The strong convergence of the random matrix towards the limit element follows from general results about the almost sure strong convergence of Gaussian and deterministic matrices \cite{male2012norm,collins2014strong}. The almost sure strong convergence implies in turn the almost sure convergence of the largest eigenvalue of $\ii J_{2n} - G_{2n} / \sqrt{2n}$ towards the supremum of the support of the limiting measure $\mu_\sigma$, i.e.~$R(\sigma)$.
\end{proof}

We plot the support of the measure $\mu_\sigma$, as a function of $\sigma$, in Figure \ref{fig:support}. Note that at $\sigma=1$ the is a \emph{phase transition} between a measure supported on two intervals (for $\sigma < 1$) and a measure supported on one interval ($\sigma \geq 1$). 

\begin{figure}[htb]
    \centering
    \includegraphics[scale=0.45]{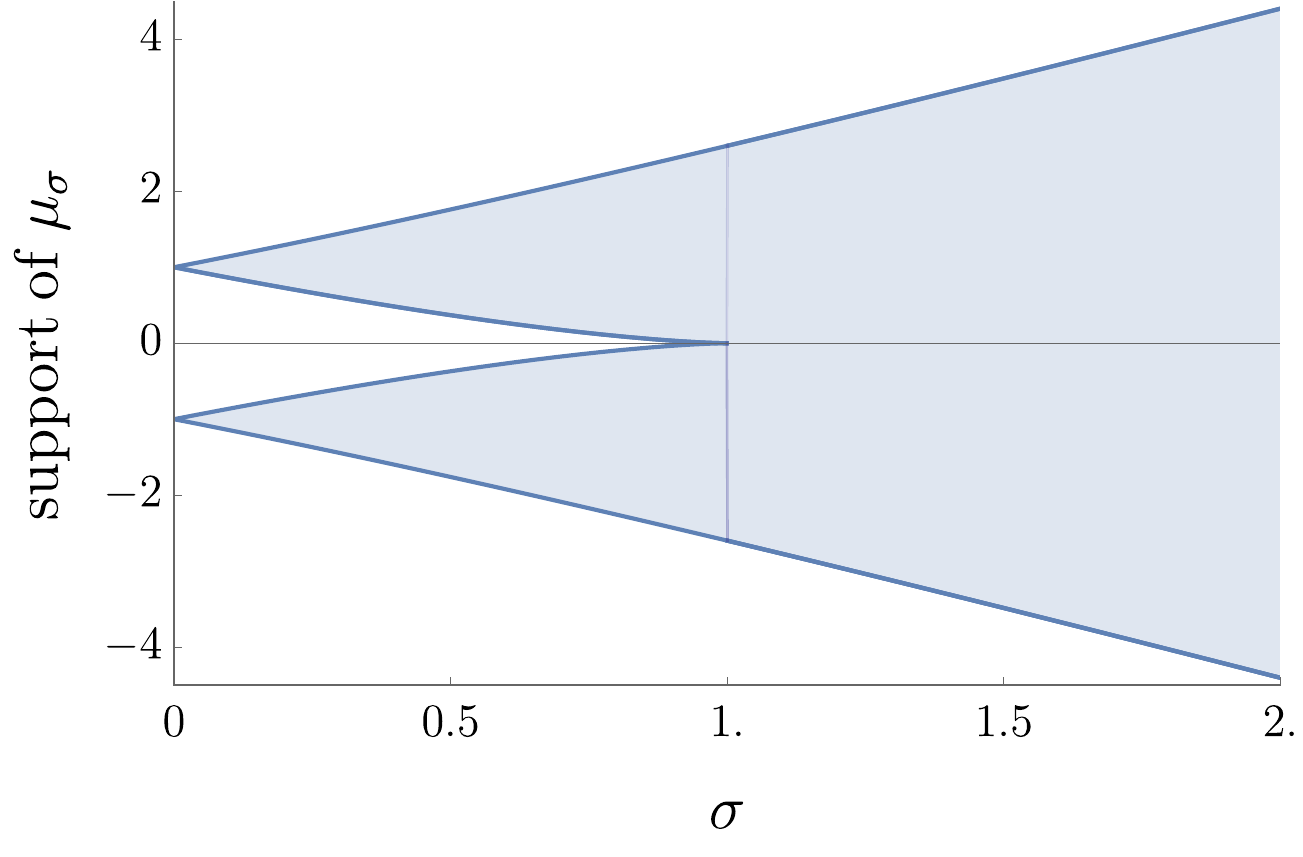}
    \caption{The support of the measure $\mu_\sigma = (1/2 \delta_{-1} + 1/2 \delta_1) \boxplus \mathrm{SC}_\sigma$ as a function of $\sigma$. For $\sigma<1$, the measure is supported on two intervals $\pm [L(\sigma), R(\sigma)]$, while for $\sigma \geq 1$, the support is the interval $[-R(\sigma), R(\sigma)].$}
    \label{fig:support}
\end{figure}

One could have obtained the result above on the support of the measure $\mu_\sigma$ using directly the machinery of Voiculescu's $\mathcal R$-transform, see \cite[Lecture 12]{nica2006lectures}. The $\mathcal R$-transform of the measure $\mu_\sigma$ is the sum of the $\mathcal R$-transforms of the Bernoulli, resp.~the semicircular distributions: 
$$\mathcal R_{\mu_\sigma}(z) = \mathcal R_{\textrm{Bernoulli}}(z) + \mathcal R_{\textrm{SC}_\sigma}(z) = \frac{\sqrt{1+4z^2}-1}{2z} + \sigma z.$$
From the $\mathcal R$-transform one can obtain a cubic equation satisfied by the Cauchy transform of the measure $\mu_\sigma$. Then, using Stieltjes' inversion formula, we can compute (intricate cubic roots) formulas for the density of the measure $\mu_\sigma$. We plot these theoretic curves against histograms of the random matrices $\ii J_{2n} - G_{n}/\sqrt{2n}$ in Figure \ref{fig:density}.

\begin{figure}[htb]
    \centering
    \includegraphics[width=0.47\textwidth]{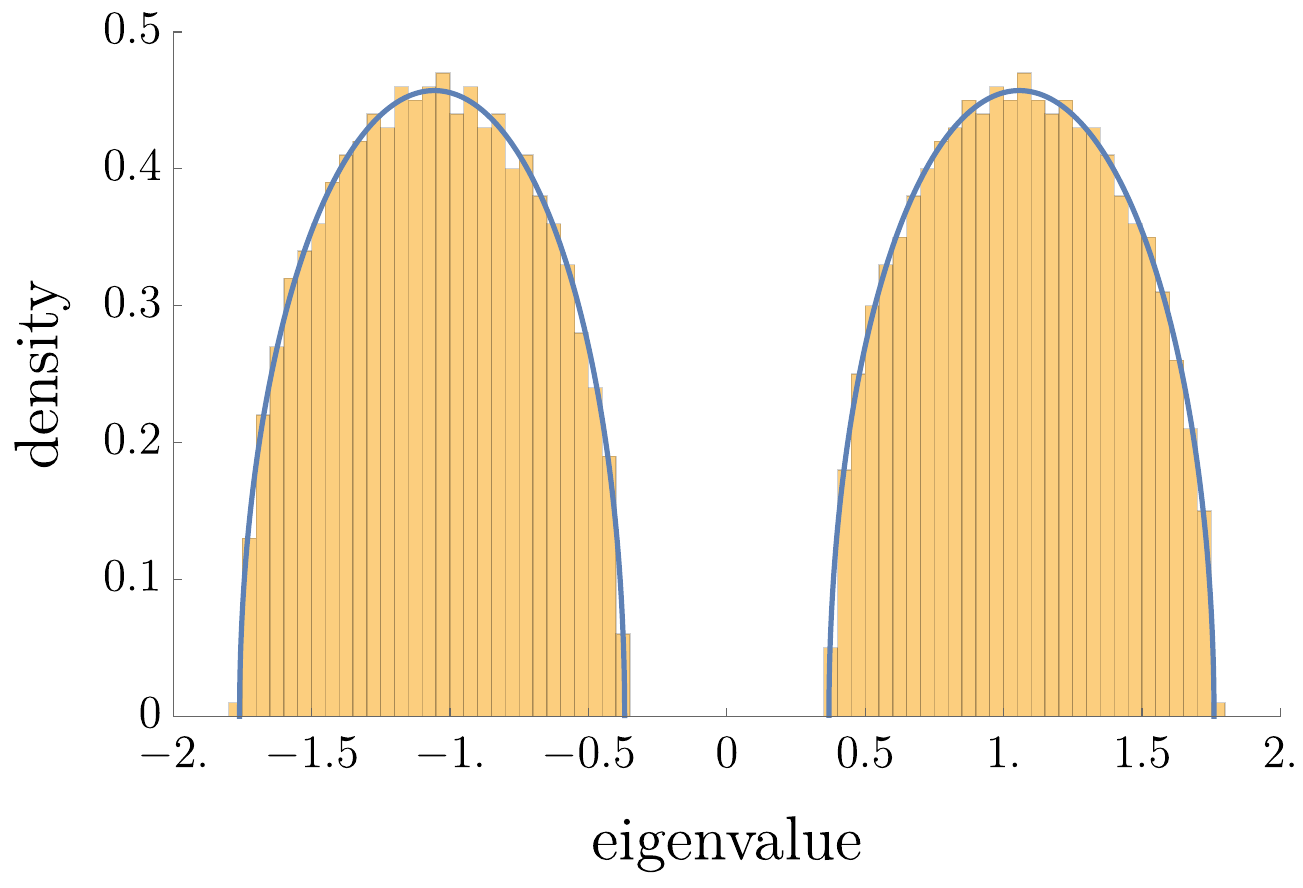} \qquad 
    \includegraphics[width=0.47\textwidth]{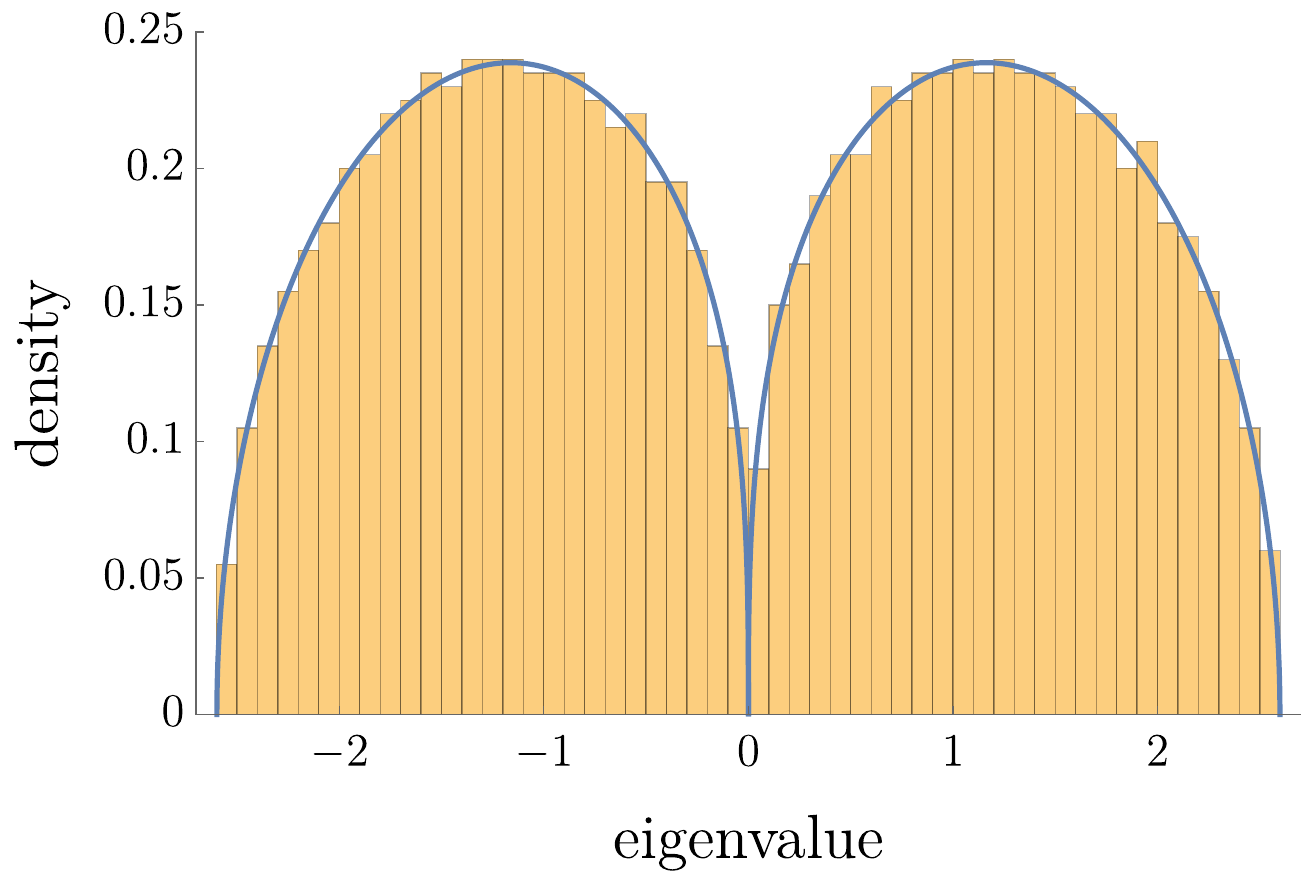} \\
    \includegraphics[width=0.47\textwidth]{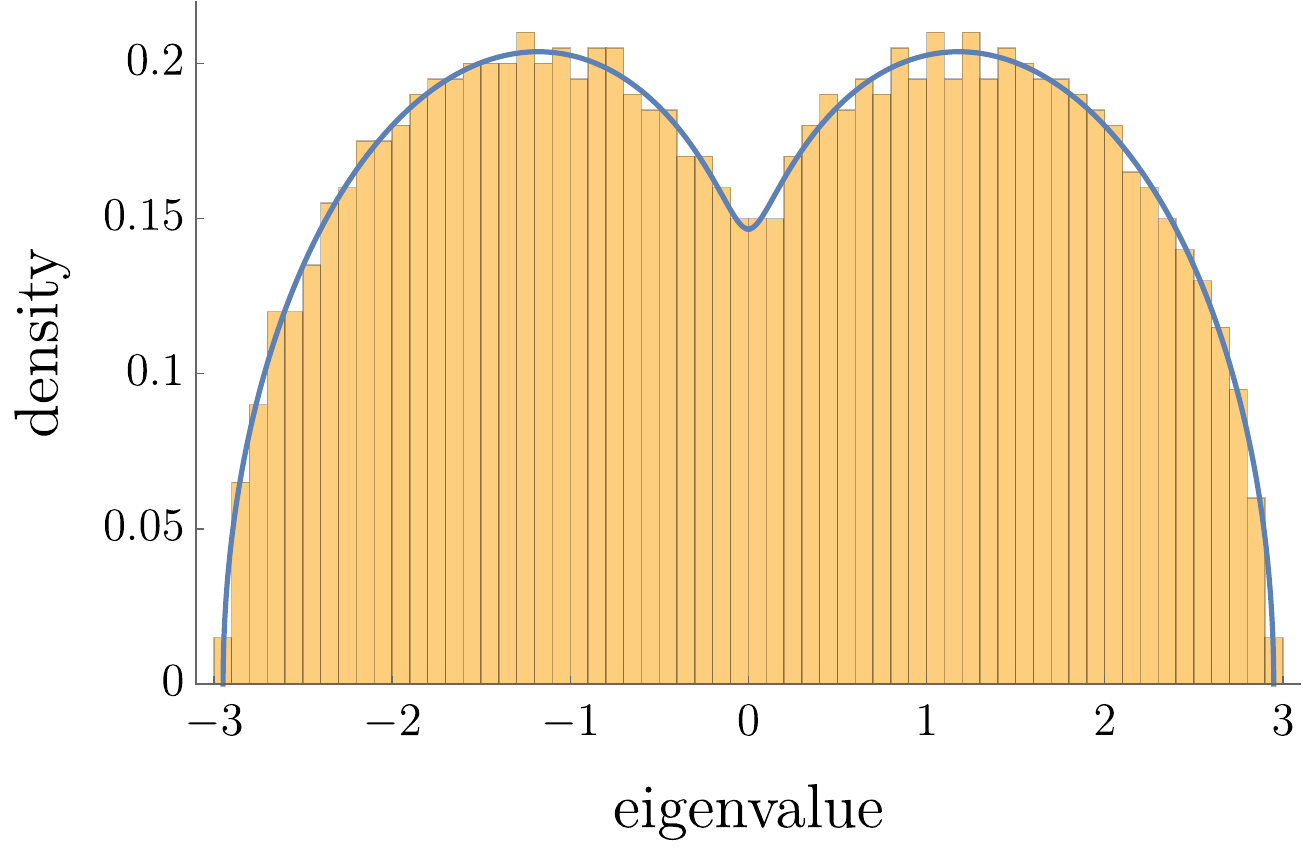} \qquad
    \includegraphics[width=0.47\textwidth]{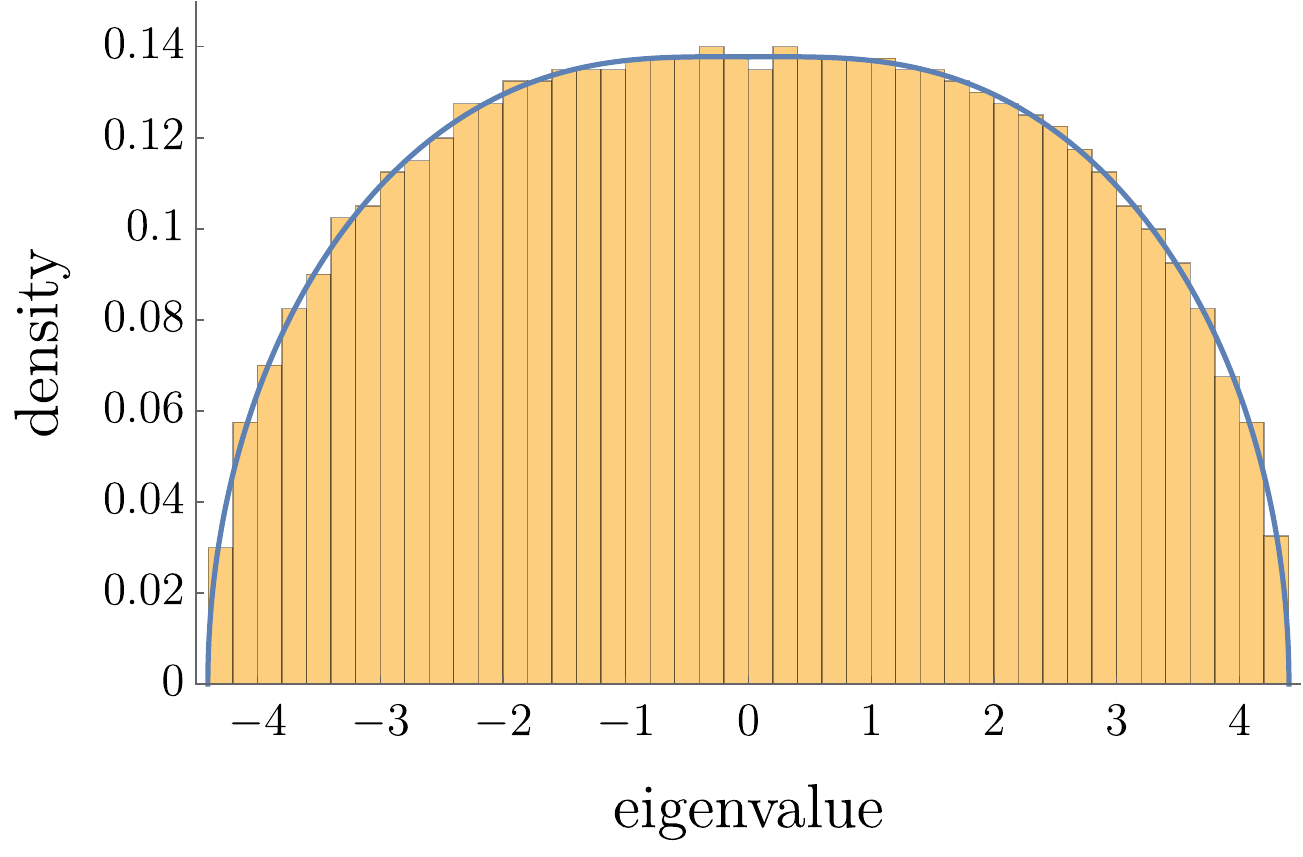}
    \caption{Theoretical limit densities (blue) vs.~empirical eigenvalue histograms (yellow) for a single realization of $\ii J_{2n} - G_{2n}/\sqrt{2n}$, with $n=1000$. From top to bottom, left to right: $\sigma = 0.5, 1, 1.2, 2$.}
    \label{fig:density}
\end{figure}

From result above, we have seen that the proper normalization of the $\GOE(2n, \sigma)$ matrix (with $\sigma >0$ fixed) is $(2n)^{-1/2}$, hence we need to consider elements of the ensemble $\RQCM(2n, \sigma / \sqrt{2n})$. 

\begin{theorem}\label{thm:RQCM-eigenvalues}
The limiting eigenvalue distribution of a random quantum covariance matrix $S_{2n} \in \RQCM(2n, \sigma/\sqrt{2n})$ is a shifted semicircular distribution $\mathrm{SC}_{R(\sigma), \sigma}$, where 
\begin{equation}\label{eq:def-SC}
    \mathrm{d} \mathrm{SC}_{m,\sigma} = \frac{\sqrt{4\sigma^2-(x-m)^2}}{2 \pi \sigma^2}  \mathbf{1}_{|x-m| \leq 2\sigma}(x) \mathrm{d}x.
\end{equation}
\end{theorem}

Let us record the behavior in $\sigma$ of the shift parameter $R$ from Eq.~\eqref{eq:def-R}:
\begin{align*}
    \text{as } \sigma \to 0, \qquad R(\sigma) &= 1 + \sqrt 2 \sigma + \frac{\sigma^2}{4} + O(\sigma^3)\\
    \text{as }\sigma \to \infty, \qquad R(\sigma) &= 2 \sigma + \frac 1 \sigma - \frac{5}{4 \sigma^3} + O(\sigma^{-5}).
\end{align*}
Note that the leading order of the expressions above match, respectively, the negative parts of the Bernoulli and semicircular distributions.

\bigskip

We now consider the \emph{partial transposition} of the random covariance matrices introduced in Definition \ref{def:RQCM-ensemble}. Recall from Eq.~\eqref{eq:PPT-condition} that a $n$-mode quantum covariance matrix $S$ satisfies the positive partial transposition (PPT) entanglement criterion with respect to the mode bipartition 
$$n = m + (n-m)$$
if the following matrix is positive semidefinite: 
$$S + \ii \begin{bmatrix} J_{2m} & 0_{2m \times 2(n-m)} \\ 0_{2(n-m) \times 2m} & J_{2(n-m)} \end{bmatrix}.$$

In the following proposition, we show that, in the limit of large $n$, random quantum covariance matrices satisfy automatically the PPT criterion. Later, in Section \ref{sec:entanglement}, we show that there exists a large proportion of PPT entangled Gaussian states.  

\begin{theorem}\label{thm:PPT}
Let $S_{2n} \in \RQCM(2n, \sigma/\sqrt{2n})$ be a normalized random quantum covariance matrix and $1 \leq m = m(n) < n$ be a parameter introducing a bipartition of the modes. Then, almost surely as $n \to \infty$, the smallest eigenvalue of the matrix 
$$S_{2n} + \ii \begin{bmatrix} J_{2m} & 0_{2m \times 2(n-m)} \\ 0_{2(n-m) \times 2m} & J_{2(n-m)} \end{bmatrix}$$
converges to 0. In particular, for all $\epsilon > 0$, the random quantum covariance matrix $S_{2n} + \epsilon I_{2n}$ satisfies the PPT criterion from Eq.~\eqref{eq:PPT-condition}.
\end{theorem}
\begin{proof}
    First, note that the deterministic matrices $J_{2n}$ and $\left[ \begin{smallmatrix} J_{2m} & 0 \\ 0 & J_{2(n-m)} \end{smallmatrix} \right]$ have the same spectrum ($\pm 1$ with multiplicity $n$). Hence, the almost sure convergence of the smallest eigenvalue of the matrix from the statement follows in the same way as in the proof of Proposition \ref{prop:convergence-shifted-GOE}. Adding an arbitrary positive number $\epsilon$ to the random quantum covariance matrix ensures the positivity as $n \to \infty$.
\end{proof}

\bigskip

The \emph{purity} $\mu$ of a Gaussian quantum state can be expressed in terms of its covariance matrix as \cite[Section 3.5]{serafini}, \cite{deGosson2019purity}
$$\mu = \frac{1}{\sqrt{ \det S}}.$$
In the case of the model of random quantum covariance matrices we study, the asymptotic behavior of the purity can be described as follows. 
\begin{prop}\label{prop:purity}
Let $S_{2n} \in \RQCM(2n, \sigma/\sqrt{2n})$ be a normalized random quantum covariance matrix. Define the following quantity (see Eqs.~\eqref{eq:def-R}, \eqref{eq:def-SC}): 
\begin{equation}\label{eq:def-LD}
    \mathrm{LD}(\sigma) := \int_{R(\sigma)-2\sigma}^{R(\sigma)+2\sigma} \log(x) \mathrm{d} \mathrm{SC}_{R(\sigma),\sigma}(x).
\end{equation}
Then, as $n \to \infty$, in probability, 
$$- \frac 1 n \log \mu(S_{2n}) \to \mathrm{LD}(\sigma).$$
In other words, the purity of a random quantum covariance matrix behaves as 
$$\mu(S_{2n}) \sim \exp (-n \mathrm{LD}(\sigma))$$
in the limit of large number of modes. 
\end{prop}
\begin{proof}
    The result is a consequence of the convergence of linear statistics of eigenvalues for the random matrix model in Proposition \ref{prop:convergence-shifted-GOE}. Note that in our case, $R(\sigma) > 2\sigma$, so the $\log$ function is bounded on the compact interval $[R(\sigma)-2\sigma, R(\sigma)+2\sigma)]$.
\end{proof}

The integral in Eq.~\eqref{eq:def-LD} is transcendental and cannot be provided in closed form. We plot in Figure \ref{fig:LD} the behavior of this function of $\sigma$. For $\sigma=1$, we obtain approximately
$$\mu(S_{2n}) \sim \exp(-0.865668 \cdot  n).$$

\begin{figure}
    \centering
    \includegraphics[scale=0.5]{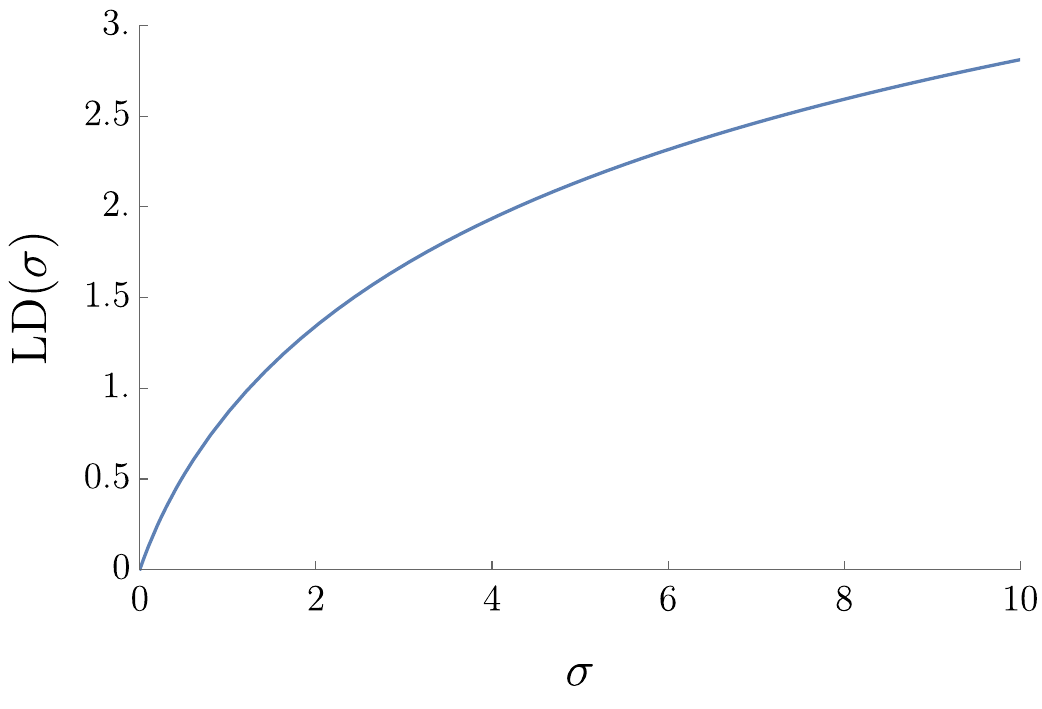}
    \caption{The rate function $\mathrm{LD}(\sigma)$ of the purity of a random quantum covariance matrix $\mu(S_{2n}) \sim \exp (-n \mathrm{LD}(\sigma))$. }
    \label{fig:LD}
\end{figure}

\bigskip

Let us now consider how the eigenvalue distribution of marginals of $\RQCM$ elements behaves in the large $n$ limit. 

\begin{prop}\label{prop:marginal-large-mn}
Let $S_{2n} \in \RQCM(2n, \sigma/\sqrt{2n})$ be a normalized random quantum covariance matrix. Consider its $m$-mode marginal $S_{2n}^{(2m)}$, where $m \to \infty$ such that $m/n \to t \in (0,1)$. Then, as $n \to \infty$, the random matrix $S_{2n}^{(2m)}$ behaves like \emph{shifted} $\RQCM(2m, \sqrt t\sigma)$ element, where the asymptotic shift is given by
$$\delta_{t, \sigma} = R(\sigma) - R(\sqrt t \sigma).$$
In particular, its limiting eigenvalue distribution is $\mathrm{SC}_{R(\sigma), \sqrt \sigma}$.
\end{prop}
\begin{proof}
    The result follows easily from Proposition \ref{prop:marginal-RQCM}:
    $$S_{2n}^{(2m)} = \Big( \frac{G_{2n}}{\sqrt{2n}} + \lambda_{\max}(\ii J_{2n} - G_{2n} / \sqrt{2n}) \Big)^{(2m)} \sim \frac{G_{2m}}{\sqrt{2m}} \frac{\sqrt{2m}}{\sqrt{2n}} + R(\sigma) - R(\sqrt t \sigma).$$
\end{proof}

\section{Symplectic eigenvalues}\label{sec:symplectic}

In this section, we compute the limiting distribution (in the large $n$ limit) of the \emph{symplectic eigenvalues} of a random Gaussian covariance matrix. Recall from \cref{sec:background} that the symplectic eigenvalues of a covariance matrix $S$ are the non-negative (usual) eigenvalues of the matrix $\ii J_{2n} \cdot S$. As in the previous section, we shall use free probability theory to compute the limiting spectrum of this matrix, in particular, the theory of $\mathcal S$-transform, see \cite[Lecture 18]{nica2006lectures}. 

We are interested in the free multiplicative product between a general semicircular distribution (with non-negative support) $\mathrm{SC}_{m,\sigma}$ and a Bernoulli distribution $\mathrm{B}:=(\delta_{-1} + \delta_1)/2$. Recall that the (shifted) semicircular distribution with mean $m$ and variance $\sigma^2$ is given by 
$$\mathrm{d} \mathrm{SC}_{m,\sigma} := \frac{\sqrt{4\sigma^2-(x-m)^2}}{2 \pi \sigma^2}  \mathbf{1}_{|x-m| \leq 2\sigma}(x) \mathrm{d}x.$$
It is supported on the interval $[m-2\sigma, m+2\sigma]$; in what follows we shall assume implicitly that $m \geq 2\sigma$, although this assumption is not needed in the $\mathcal S$-transform computations. In free probability theory, the $\mathcal S$-transform has the property that, given two non-commutative random variables $x,y$, with $x \geq 0$, 
$$\mathcal S(\sqrt x y \sqrt x) = \mathcal S(x) \cdot \mathcal S(y).$$
It is a complex variable function defined by 
$$\mathcal S(z)  = \frac{1+z}{z} \chi(z),$$
where $\chi(\cdot)$ is the functional inverse of the moment generating function $\psi$
$$\psi(z) = \sum_{n=1}^\infty z^n m_n,$$
where $m_n$ are the moments of the corresponding non-commutative random variable. Direct computations yield
\begin{align*}
    \mathcal S_{\mathrm{B}}(z) &= \frac{\sqrt{1+z}}{\sqrt z}\\
    \mathcal S_{\mathrm{SC}_{m, \sigma}}(z) &= \frac{2}{m+\sqrt{m^2+4\sigma^2z}},
\end{align*}
which yield 
$$\mathcal S_{\mathrm{B} \boxtimes \mathrm{SC}_{m, \sigma}}(z) = \frac{2\sqrt{1+z}}{\sqrt z(m+\sqrt{m^2+4\sigma^2z})}.$$
We now follow the inverse procedure, extracting the distribution $\mathrm{B} \boxtimes \mathrm{SC}_{m, \sigma}$ from its $\mathcal S$-transform. We find that the Stieltjes transform $G(z) = (1+\psi(1/z))/z$ of the free product distribution satisfies the algebraic equation 
$$\frac{2\sqrt{zG-1}}{\sqrt{zG}(m+\sqrt{m^2+4\sigma^2(zG-1)})} = \frac 1 z.$$
After some simplifications, we find that $G$ satisfies the cubic equation 
$$z\sigma^2 G^3 -\sigma^2(2z^2+\sigma^2)G^2+z(z^2+2\sigma^2-m^2)G -z^2=0.$$
We summarize these computations in the result below. 
\begin{theorem}\label{thm:symplectic-eigenvalues}
    The limiting \emph{symplectic} eigenvalue distribution of a random quantum covariance matrix $S_{2n}\in \RQCM(2n, \sigma/\sqrt{2n})$ is the non-negative part of the probability measure $\mathrm{B} \boxtimes \mathrm{SC}_{R(\sigma), \sigma}$, a free multiplicative convolution of a Bernoulli distribution and a non-centered semicircular distribution. The Stieltjes transform of this distribution satisfies the cubic equation 
    $$z\sigma^2 G^3 -\sigma^2(2z^2+\sigma^2)G^2+z(z^2+2\sigma^2-R(\sigma)^2)G -z^2=0.$$
    The support of the non-negative part of this distribution is the interval $[1, \sqrt{F(\sigma)}]$, with 
    \begin{align*}
    F(\sigma) &:= \frac{1}{2}-\frac{\sigma ^4}{8}+4 \sigma ^2+\sqrt{\sigma ^2+8} \sigma +\frac{1}{8} \sqrt{\sigma ^2+8} \sigma ^3+ \\
    &\qquad \frac{1}{64} \left(-4096 \sigma ^6+78336 \sigma ^4+49152 \sigma ^2+ \right.\\
    &\qquad \left.4096 \sqrt{\sigma ^2+8} \sigma +4096 \sqrt{\sigma ^2+8} \sigma ^5+33280 \sqrt{\sigma ^2+8} \sigma ^3+1024 \right)^{1/2}.
    \end{align*}
\end{theorem}
\begin{proof}
    The only statement to be proven is the formula for the right end of the support of $\mathrm{B} \boxtimes \mathrm{SC}_{R(\sigma), \sigma}$. Computing the discriminant of the cubic equation satisfied by $G$ and solving for $z$ yields a set of $6$ solutions: $z = \pm 1$, two complex solutions, and two opposite real solutions. The formula in the statement corresponds to largest positive real solution. 
\end{proof}
\begin{rem}
    The behavior of the right edge of the support $\sigma \mapsto \sqrt{F(\sigma)}$ is surprisingly close to a linear function. The asymptotic behavior in the limiting cases is given by
    $$\text{as $\sigma \to 0$: } \sqrt{F(\sigma)} \sim  1+2 \sqrt{2} \sigma  \qquad \text{ and } \qquad \text{as $\sigma \to \infty$: } \sqrt{F(\sigma)} \sim \sqrt{\frac{11}{2}+\frac{5 \sqrt{5}}{2}} \sigma .$$
\end{rem}

Writing down the cubic equation solutions is cumbersome, so we decide here to focus on some special cases and on the support of the measure $\mathrm{B} \boxtimes \mathrm{SC}_{m, \sigma}$. For example, after replacing $m=R(\sigma)$ with the value from Eq.~\eqref{eq:def-R} and taking $\sigma = 1$, we find that the symplectic eigenvalues of $S_{2n}$ have, in the limit $n \to \infty$, density
$$\frac{-4 x^4-73 x^2+\left(-8 x^6+510 x^4+3 \left(9 \sqrt{-16 x^6+264 x^4-237 x^2-11}+73\right) x^2+8\right)^{2/3}-4}{2 \sqrt{3} \pi  x \sqrt[3]{-8 x^6+510 x^4+3 \left(9 \sqrt{-16 x^6+264 x^4-237 x^2-11}+73\right) x^2+8}}$$
supported on the interval 
$$x \in \left[ 1, \sqrt{\frac{9 \sqrt{3}}{2}+\frac{31}{4}}\right] \approx [1,3.94262].$$
We display this density along with numerical experiments in Figure \ref{fig:symplectic-eigenvalues}, left panel; in the right panel we plot the analytical curve along with the histogram in the case $\sigma = 10$.

\begin{figure}[t]
    \centering
    \includegraphics[width=0.47\textwidth]{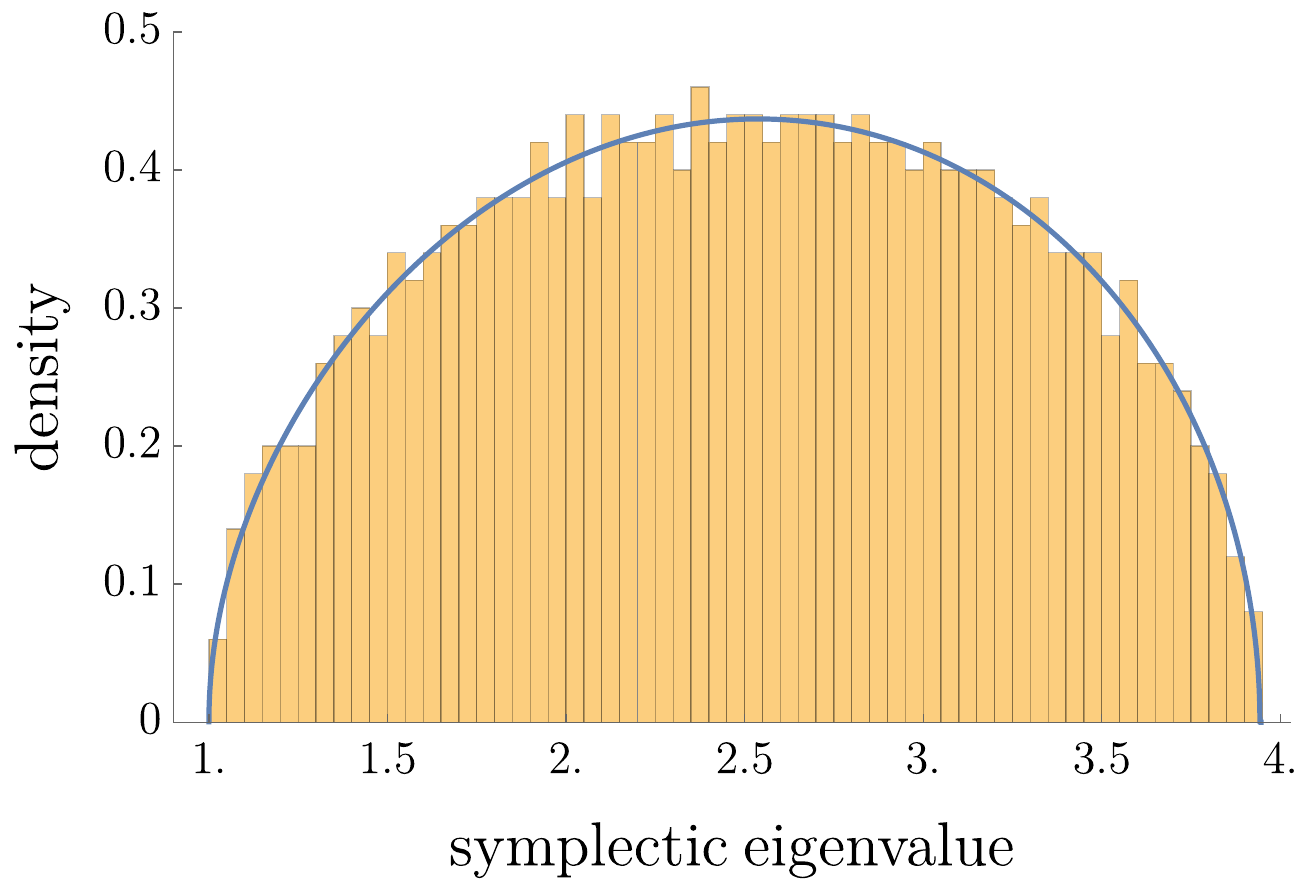} \qquad \includegraphics[width=0.47\textwidth]{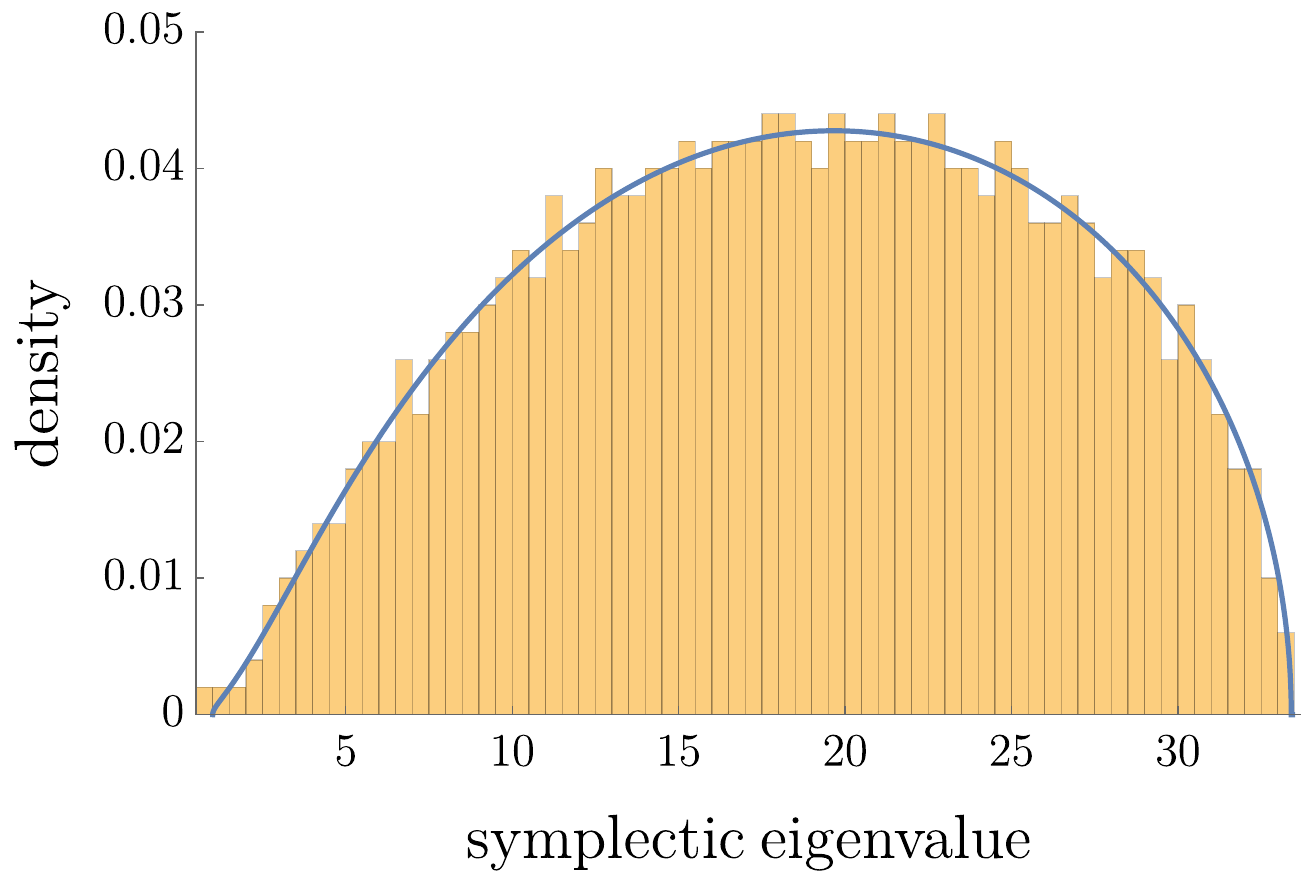}
    \caption{Numerical histograms for the symplectic eigenvalues of a random quantum covariance matrices for $n=1000$ versus theoretical densities. Left panel: $\sigma=1$; right panel: $\sigma = 10$.}
    \label{fig:symplectic-eigenvalues}
\end{figure}

Concerning the physical interpretation of the above result, the largest symplectic eigenvalue $\sqrt{F(\sigma)}$ can be understood at an \emph{energy upper bound} for the random quantum covariance matrix. It would be interesting to derive an explicit formula for the average \emph{energy per mode} of the Gaussian state, given in the limit $n \to \infty$ as the integral 
$$\int |x| \mathrm d(\mathrm{B} \boxtimes \mathrm{SC}_{R(\sigma), \sigma})(x).$$
Given that a tractable formula for the density of the limiting measure is unavailable in the general case, we can only compute it for special values of $\sigma$. In the case $\sigma=1$, we obtain an approximate value $2.49289$ for the energy per mode, to be compared wit the maximal energy $3.94262$ for the same value of $\sigma$.

\section{Entanglement and extendability of random quantum covariance matrices}\label{sec:entanglement}

In the last part of our investigation we look at the entanglement and extendability properties of the matrices in the set $\RQCM(2n, \sigma)$. As was reviewed in Section \ref{sec:background}, a Gaussian state in a bipartite state is separable if and only if its covariance matrix is completely extendable (Theorem \ref{th:2}). The complete extendability criterion is a semi-definite program (SDP) which we can use to numerically investigate the proportions of separable and entangled Gaussian states corresponding to the random quantum covariance matrices. Since the entanglement problem is formulated as an SDP, it renders it hard to analytically analyze. It is for this reason that we resort in this section to a numerical analysis, leaving analytical results for future investigation. For comparison we also explore the PPT criterion Eq. \eqref{eq:PPT-condition} and numerically analyze the proportion of PPT and non-PPT random quantum covariance matrices.

Furthermore, for entangled Gaussian states we can do a finer numerical analysis to see up to which value of $k$ the corresponding random quantum covariance matrices are $k$-extendable. As pointed out in Theorem \ref{th:3} we can take this maximum value of $k$ to depict the level of entanglement: the lower this maximum value of $k$ is the further away from complete extendability, or equivalently, from the set of separable states the covariance matrix is. Theorem \ref{th:3} also shows that the problem of checking $k$-extendability for some fixed value of $k$ can also be cast as an SDP which allows us to perform our numerical analysis described above. 

We note that in the SDPs when determining whether a random quantum covariance matrix represents a separable state or not or whether it is $k$-extendible or not we must do it within some numerical precision which we must choose. For example, we used the precision of $10^{-8}$ so that every RQCM that is within that range of being determined separable is labeled as separable. Thus, for individual samples the labels are not completely error-free but since we have used the same precision throughout our calculations, it still gives us a clear idea of the entanglement and extendibility properties of the set $\RQCM(2n, \sigma)$ as a whole. We also noticed that in certain cases (with very large $\sigma$ for example) the samples seem to be drawn from some kind of boundary of the set so in these cases the error caused by the numerical precision is even larger. In our calculations we tried to avoid such cases to the best of our abilities.

The numerical calculations were performed by using a \textsf{Mathematica} notebook containing numerical routines to sample random quantum covariance matrices and to test the extendability properties of the generated random quantum covariance matrices. The notebook is available at \cite{notebook}.

\subsection{Varying the total number of modes with even partition of subsystems}

Our first analysis was done in the case when the matrices in $\RQCM(2n,\sigma)$ with $\sigma=1$ are assumed to correspond to bipartite states with $n=m+m$ modes for various values of $m$. In particular, we chose the cases when $n=10$, $n=20$, $n=50$ and $n=100$ in which cases we could still obtain decent sample sizes and run the extendability SDPs in reasonable time. The portions of separable/entangled states and PPT/non-PPT states are represented in \cref{fig:m+m-sep-ppt}. For the entangled states we did the finer analysis of determining for each sample what is the maximum $k$ such that the matrix is $k$-extendable and our findings can be found in \cref{fig:m+m}. 

\begin{figure}[htb]
    \centering
    \includegraphics[width=0.4\textwidth]{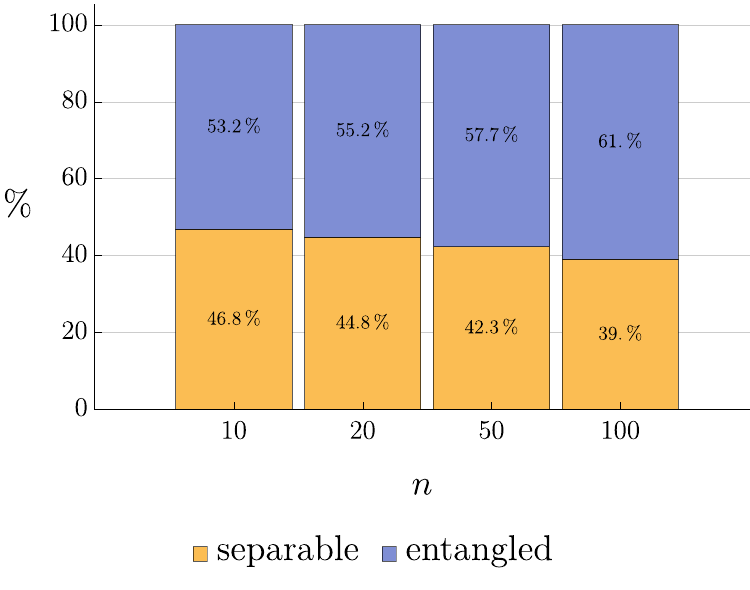} \qquad \includegraphics[width=0.4\textwidth]{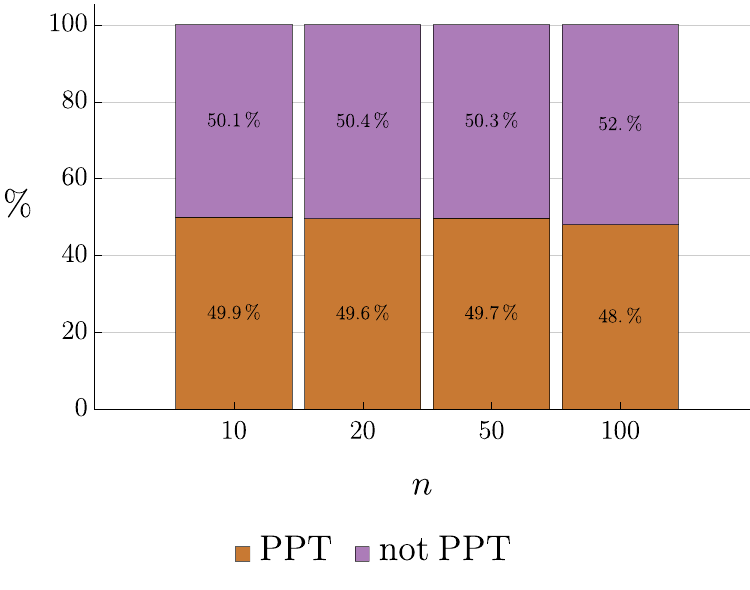} 
    \caption{Numerical histograms in the case of $n=m+m$ modes for the proportion of entangled/separable states (left) and PPT/non-PPT states (right) when $\sigma=1$. The sample sizes were 50 000 for $n=10$, 10 000 for $n=20$, 5 000 for $n=50$ and $400$ for $n=100$.}
    \label{fig:m+m-sep-ppt}
\end{figure}

\begin{figure}[htb]
    \centering
    \includegraphics[width=0.4\textwidth]{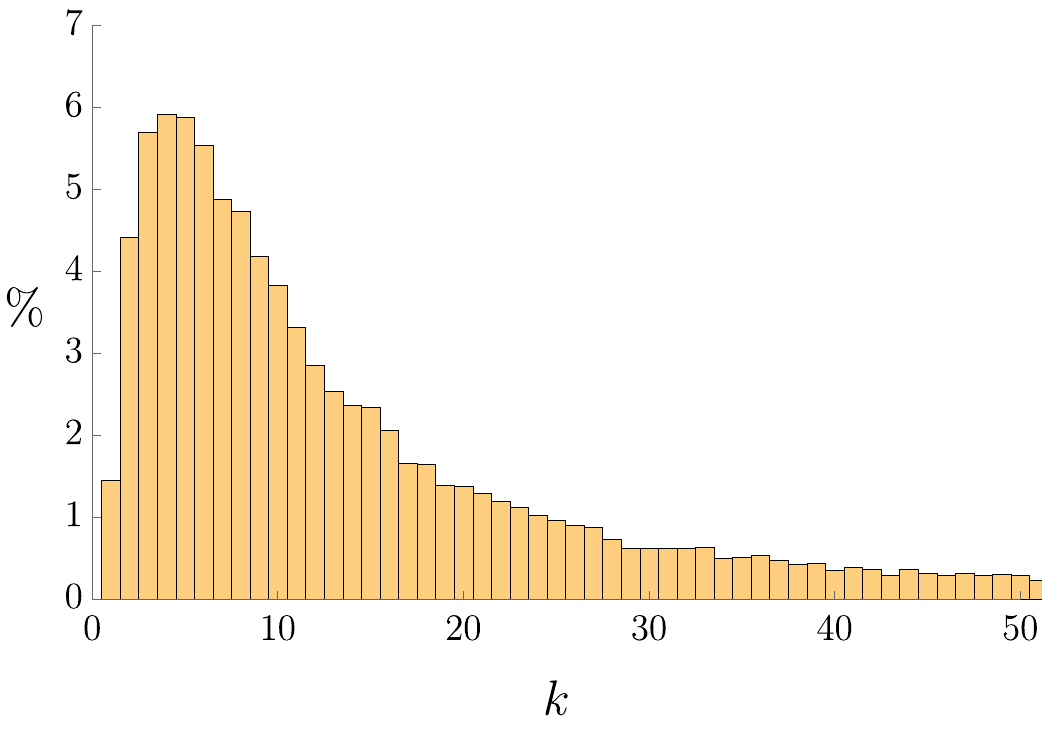} \qquad \includegraphics[width=0.4\textwidth]{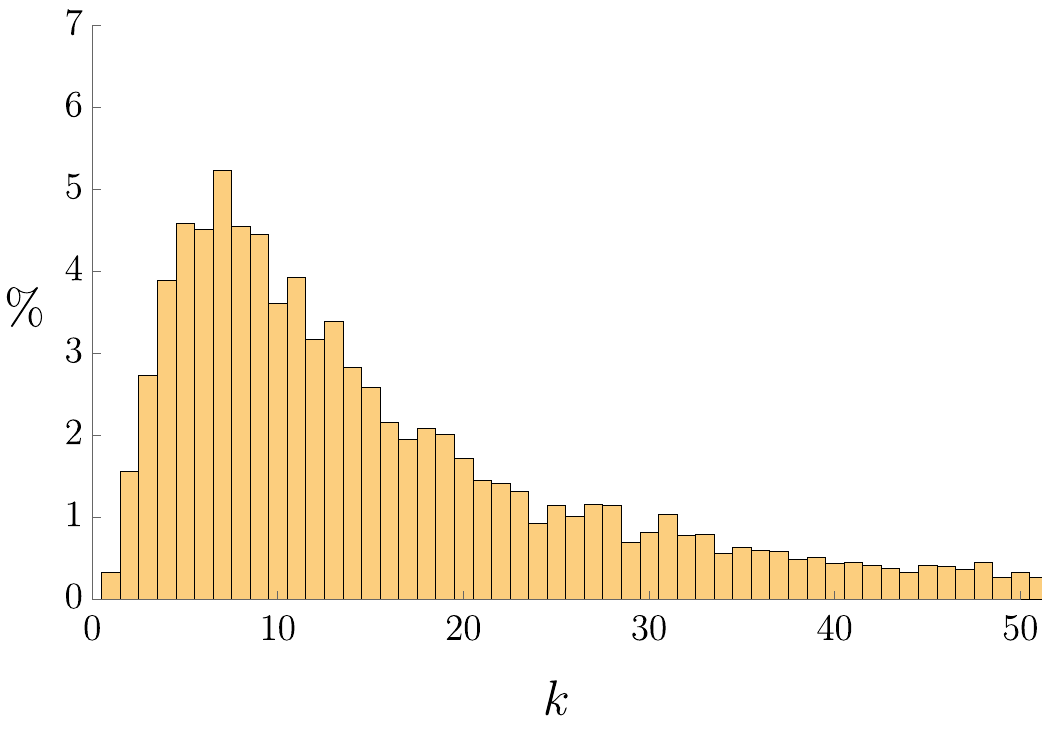} \\
    \includegraphics[width=0.4\textwidth]{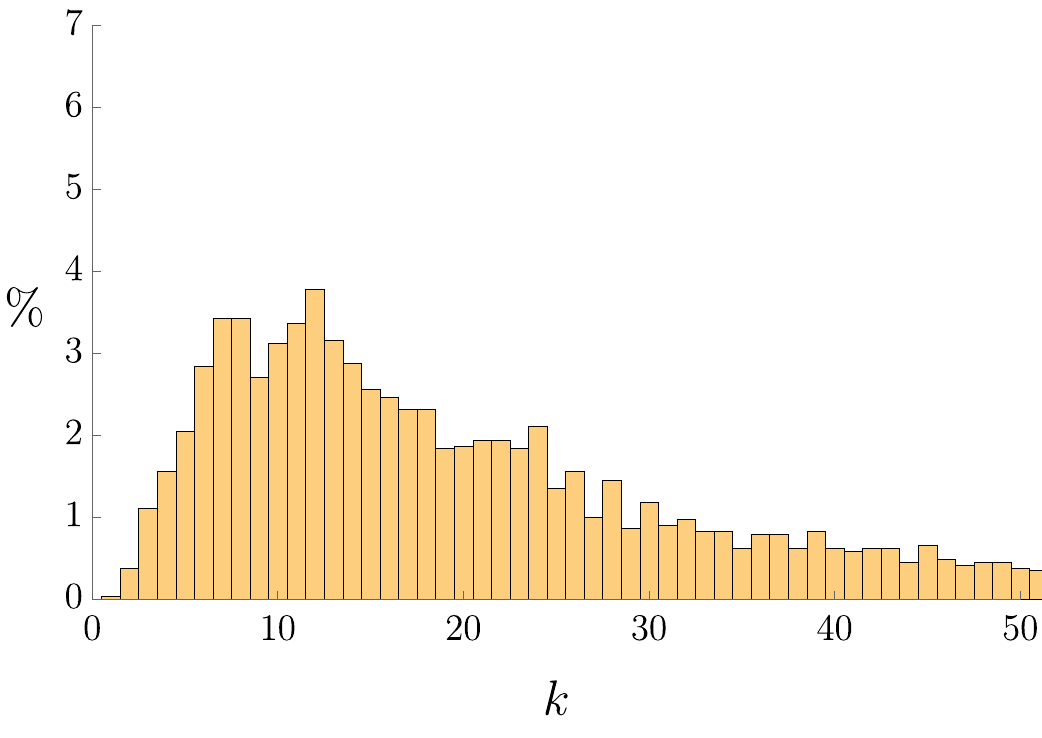} \qquad \includegraphics[width=0.4\textwidth]{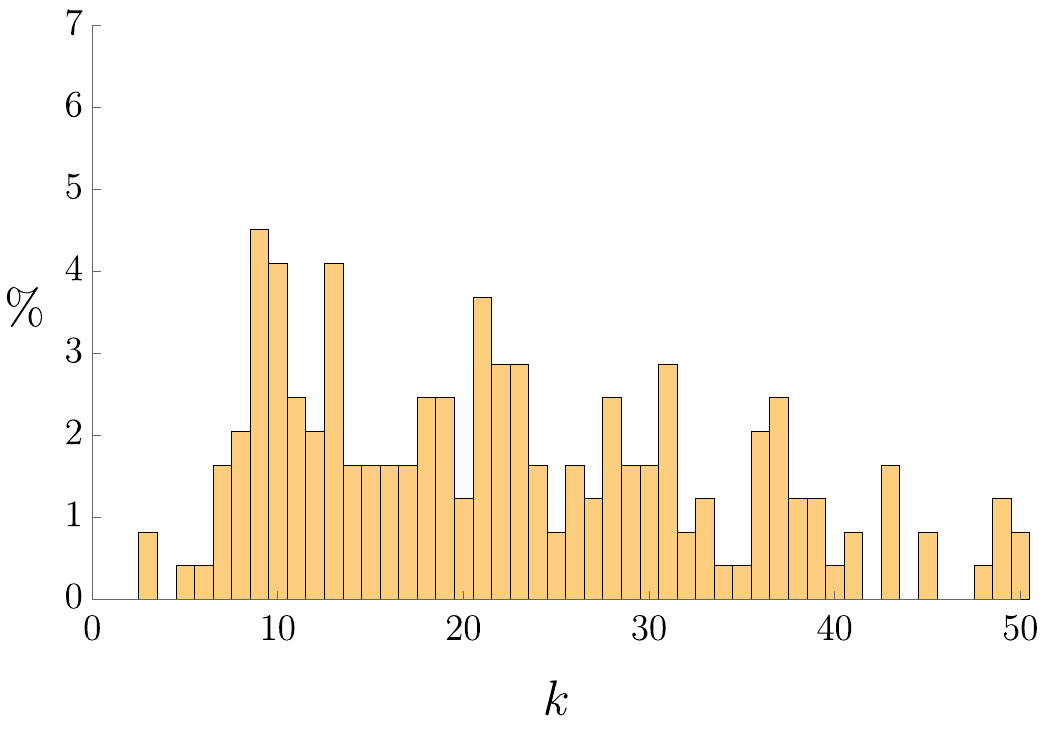}
    \caption{Numerical histograms in the case of $n=m+m$ modes ($m=5$ top left, $m=10$ top right, $m=25$ bottom left and $m=50$ bottom right) for the distribution of maximum values of $k$ for each entangled random quantum covariance matrix such that the matrix is $k$-extendable. The percentage is calculated from the total number of entangles states. The sample sizes were 50 000 for $n=10$, 10 000 for $n=20$, 5 000 for $n=50$ and $400$ for $n=100$.}
    \label{fig:m+m}
\end{figure}

From Fig. \ref{fig:m+m-sep-ppt} it would seem that with an even partition of subsystems, the portion of separable states decreases as we increase the total number of modes. However, for all number of modes, roughly half of the samples are PPT and the proportion of PPT states vs. non-PPT states seems not to be affected by the increase of the number of modes. Further numerical analysis for the minimal eigenvalues of the PPT condition (Eq. \eqref{eq:PPT-condition}) for the random quantum covariance matrices shows that for all different total number of modes (with equal partitioning of subsystems) the minimal eigenvalues seem to have a Gaussian distribution around zero. However, we find that the variance of the distribution decreases as the total number of modes increases. This is supported by Theorem \ref{thm:PPT} which states that in the limit $n \to \infty$ the minimum eigenvalue of the PPT condition (Eq. \eqref{eq:PPT-condition})  for a random quantum covariance matrix converges to zero. Thus, at fixed $n$ the minimum eigenvalues have fluctuations around its average but as $n$ increases the fluctuations decrease and in the limit $n \to \infty$ the fluctuations disappear making all the states almost surely PPT.

 Although a higher portion of states are entangled for increasing number of modes, from Fig. \ref{fig:m+m} we see that as the total number of modes increases, the entangled states become less entangled, i.e., they are closer to the set of separable states. However, since we are only looking at even partitions, so that as the total number of modes increases also the subsystem sizes increase, we cannot conclude from this analysis whether the increase of entangled states and the decrease in the level of entanglement results from the increase of the total number of modes or from the increase of the subsystem sizes. Thus, next we will consider the case when we fix one of the subsystem sizes and increase the total number of modes.

\subsection{Varying the total number of modes with fixed subsystem size} \label{subsec:even-split}

For our subsequent analysis we chose to fix the second subsystem to consist of only two modes; for $m+1$ modes it is known that the PPT criterion Eq. \eqref{eq:PPT-condition} detects all the separable states.

In particular, we looked at the case when we have a total of $n=m+2$ modes for the cases when $n=10$, $n=20$, $n=50$ and $n=100$. We kept the choice $\sigma =1$ as before. The portions of separable/entangled states and PPT/non-PPT states are represented in Fig. \ref{fig:m+2-sep-ppt}. For the entangled states we did the finer analysis of determining for each sample what is the maximum $k$ such that the matrix is $k$-extendable and our findings can be found in Fig. \ref{fig:m+2}. 

\begin{figure}[htb]
    \centering
    \includegraphics[width=0.43\textwidth]{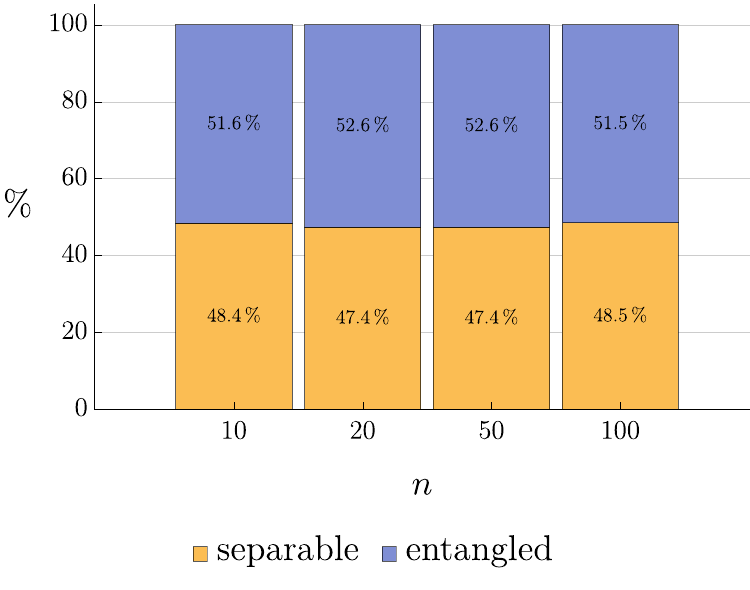} \qquad \includegraphics[width=0.43\textwidth]{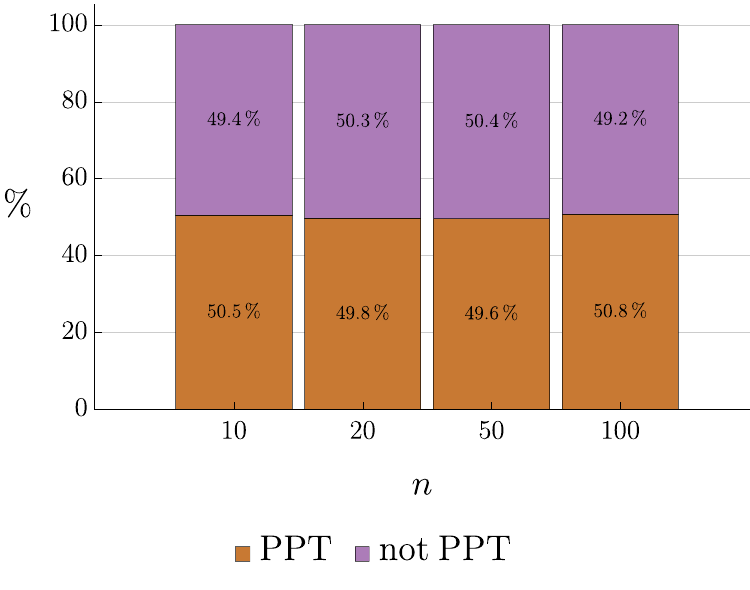} 
    \caption{Numerical histograms in the case of $n=m+2$ modes for the proportion of entangled/separable states (left) and PPT/non-PPT states (right) when $\sigma=1$ and $n \in \{10,20,50,100\}$. For all $n$ we used 5 000 samples.}
    \label{fig:m+2-sep-ppt}
\end{figure}

\begin{figure}[htb]
    \centering
    \includegraphics[width=0.4\textwidth]{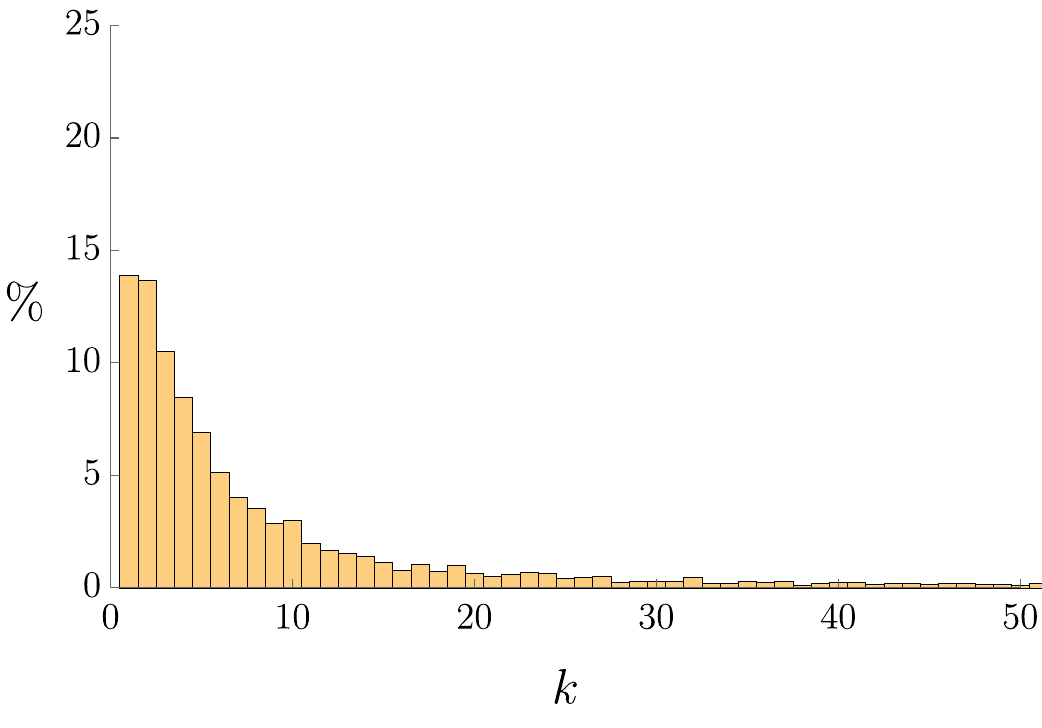} \qquad \includegraphics[width=0.4\textwidth]{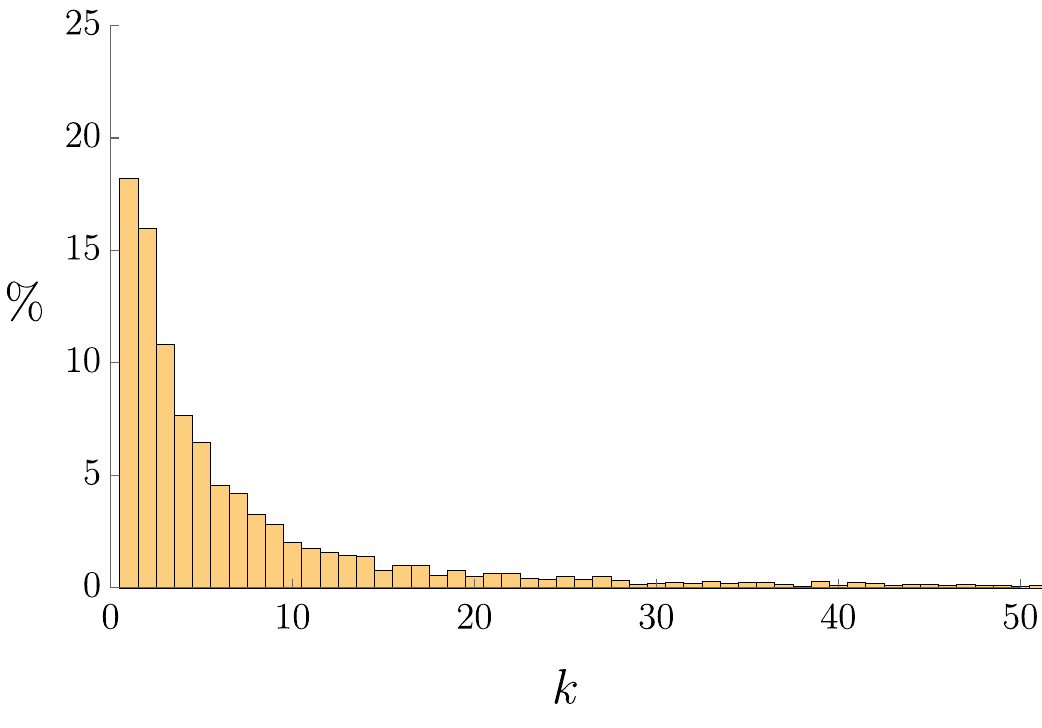} \\
    \includegraphics[width=0.4\textwidth]{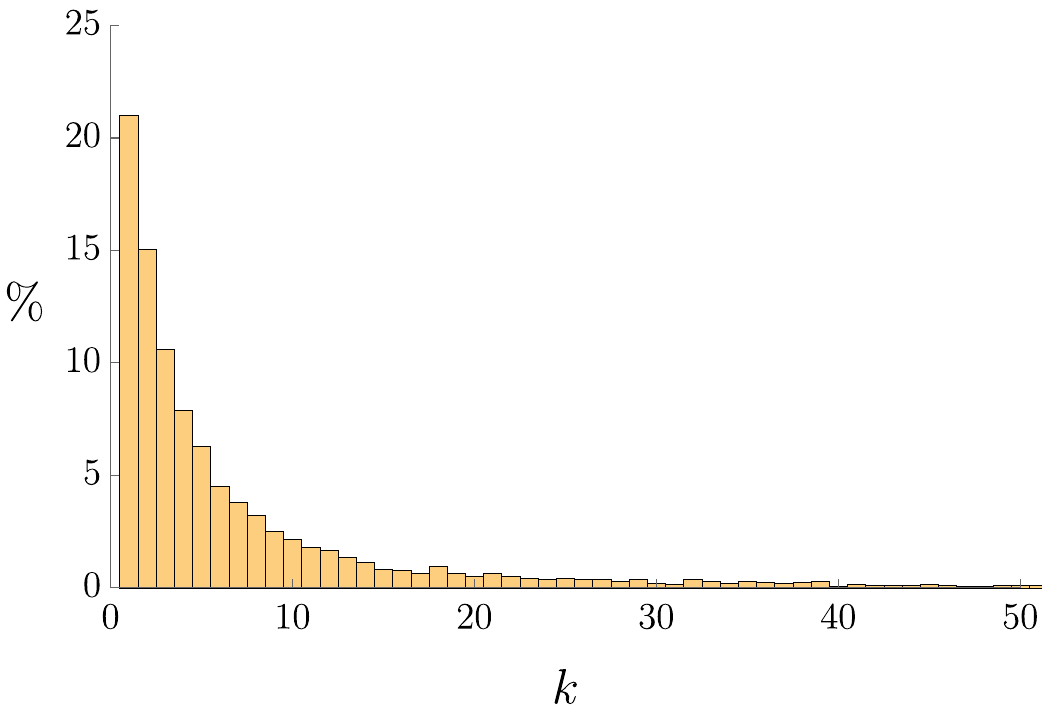} \qquad \includegraphics[width=0.4\textwidth]{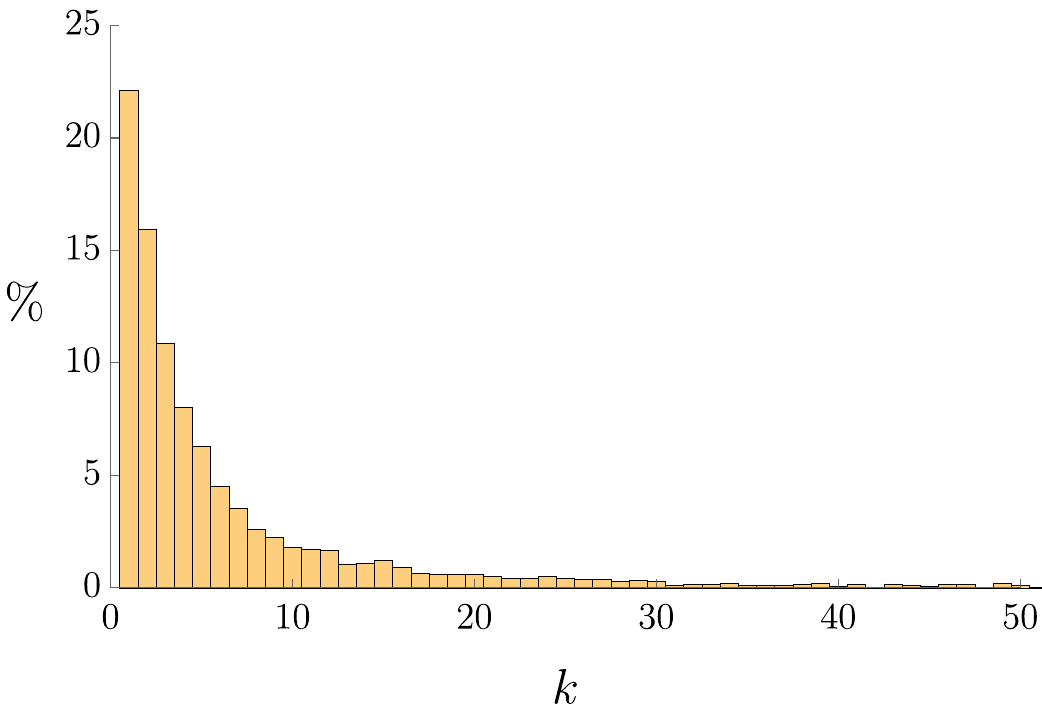}
    \caption{Numerical histograms in the case of $n=m+2$ modes ($n=10$ top left, $n=20$ top right, $n=50$ bottom left and $n=100$ bottom right) for the distribution of maximum values of $k$ for each entangled random quantum covariance matrix such that the matrix is $k$-extendable. The percentage is calculated from the total number of entangles states. For all $n$ we used 10 000 samples.}
    \label{fig:m+2}
\end{figure}

Now we see from Fig. \ref{fig:m+2-sep-ppt} that when the second subsystem size is fixed to two modes then the number of separable states remains roughly the same even when the total number of modes is increased. This would suggest that the increase in the number of entangled states which we witnessed in Sec. \ref{subsec:even-split} in the case of even subsystem partition is more of a result from the increase of the size of one of the subsystems rather than the increase of the total number of modes. For the proportion of PPT/non-PPT states we find a similar behaviour as in the previous case: roughly half of the samples are PPT states and even a further numerical analysis on the minimal eigenvalue distributions of the PPT condition (\cref{eq:PPT-condition}) for the random quantum covariance matrices again confirms this.

On the other hand, from Fig. \ref{fig:m+2} we see that when the second subsystem is fixed to have precisely two modes then the level of entanglement of the entangled states increases as the total number of modes, and thus the size of the first subsystem, increases. This is contrary to Fig. \ref{fig:m+m} in the case of even partition of subsystems which leads us to believe that the level of entanglement of the entangled random quantum covariance matrices is more proportional to the size of the first subsystem rather than the total number of modes. In order to see more evidence for this we will next fix the total number of modes and merely vary the partition of the subsystems.

\subsection{Varying the partition size with fixed number of total modes}\label{subsec:uneven-split}

Next we looked at the case when we have a fixed number of total modes $n=m+l$, where we chose  $n=20$ (for convenience), and considered the cases when the number of modes of the first subsystem are $m=2$, $m=10$ and $m=18$. Again we kept $\sigma=1$. The portions of separable/entangled states and PPT/non-PPT states are represented in Fig. \ref{fig:n=20-sep-ppt}. For the entangled states we did the finer analysis of determining for each sample what is the maximum $k$ such that the matrix is $k$-extendable and our findings can be found in Fig. \ref{fig:n=20}. 

\begin{figure}[htb]
    \centering
    \includegraphics[width=0.45\textwidth]{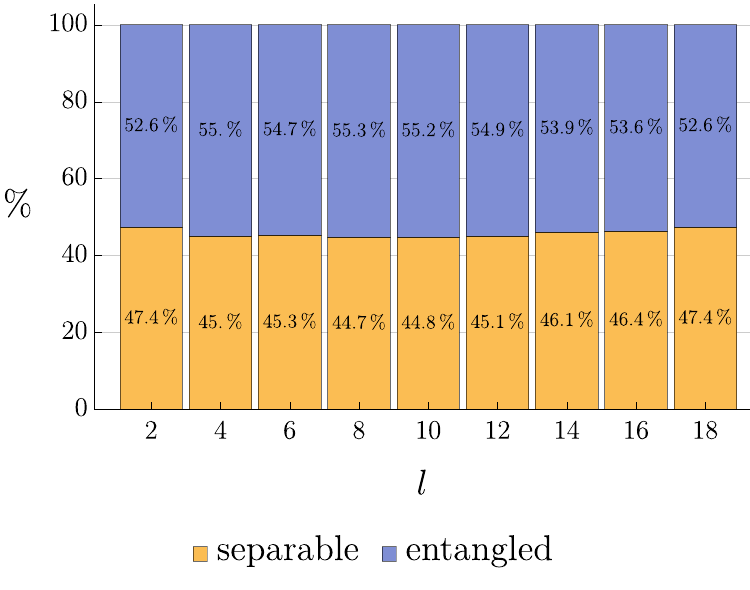} \qquad \includegraphics[width=0.45\textwidth]{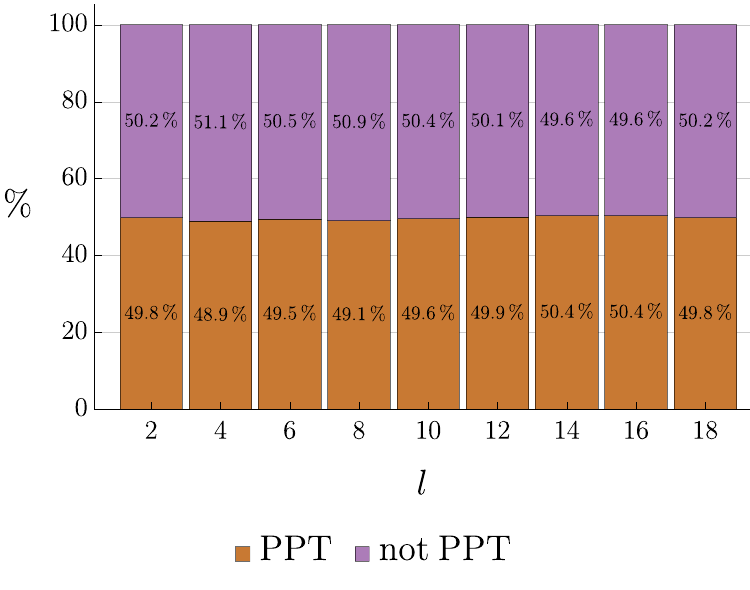} 
    \caption{Numerical histograms in the case of $n=m+l=20$ modes for the proportion of entangled/separable states (left) and PPT/non-PPT states (right) when $\sigma=1$ and $l \in \{2,4,6,8,10,12,14,16,18\}$. For all subsystem partitions we used the same 10 000 samples.}
    \label{fig:n=20-sep-ppt}
\end{figure}

\begin{figure}[htb]
    \centering
    \includegraphics[width=0.31\textwidth]{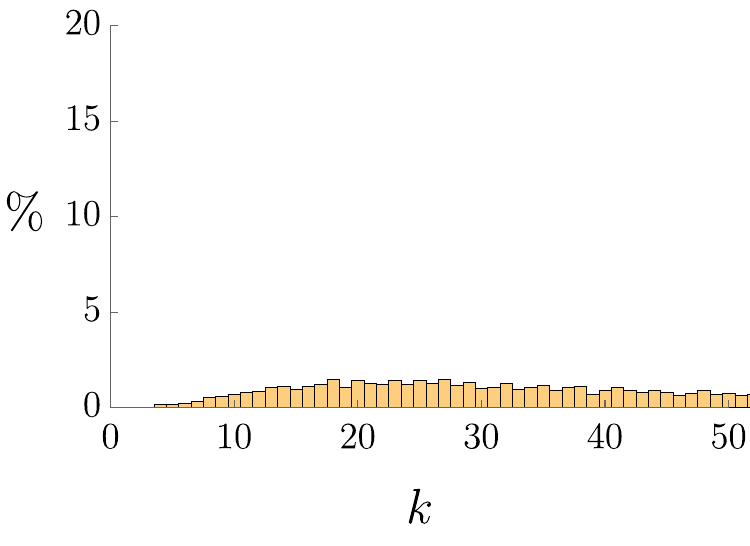} \quad
    \includegraphics[width=0.31\textwidth]{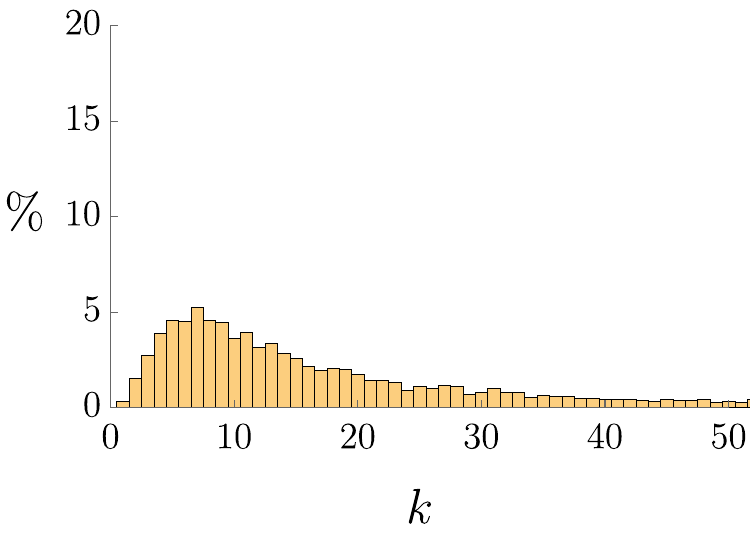} \quad
    \includegraphics[width=0.31\textwidth]{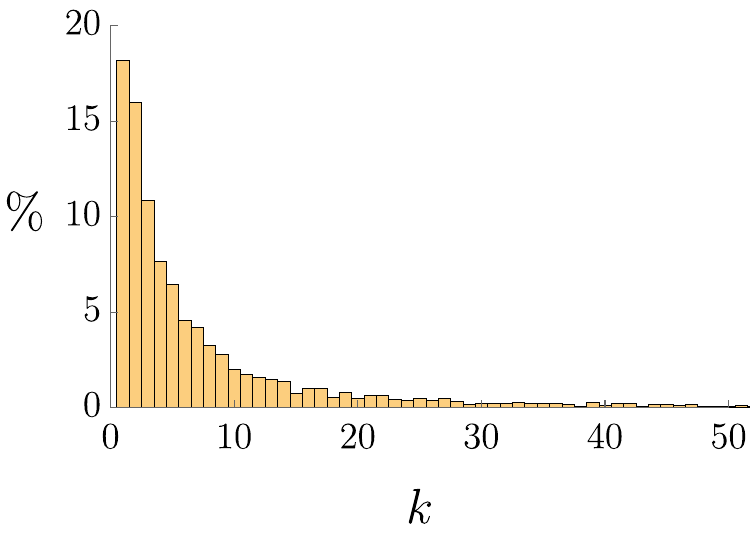}
    \caption{Numerical histograms in the case of $n=m+l=20$ modes ($m=2$ left, $m=10$ center and $m=18$ right) for the distribution of maximum values of $k$ for each entangled random quantum covariance matrix such that the matrix is $k$-extendable. The percentage is calculated from the total number of entangles states. For all subsystem partitions we used the same 10 000 samples.}
    \label{fig:n=20}
\end{figure}

By looking at Fig. \ref{fig:n=20-sep-ppt} we find that as one of the subsystem sizes increases, then the number of separable states increases. This seems to support the claim we made earlier in Sec. \ref{subsec:uneven-split}: it is the increase of the size of one of the subsystems rather than the increase of the total number of modes which results in the increase of the number of the entangled states. Also Fig. \ref{fig:n=20} seems to drastically support our previous conclusion that it is the size of the first subsystem which the level of entanglement of the entangled states is proportional to. For the proportion of PPT/non-PPT states we again find a similar behaviour as in the previous cases: roughly half of the samples are PPT states and even a further numerical analysis on the minimal eigenvalue distributions of the PPT condition (\cref{eq:PPT-condition}) for the random quantum covariance matrices confirms this.

\subsection{\texorpdfstring{Varying the standard deviation $\sigma$}{Varying the standard deviation sigma}}

\begin{figure}[h]
    \centering
    \includegraphics[width=0.43\textwidth]{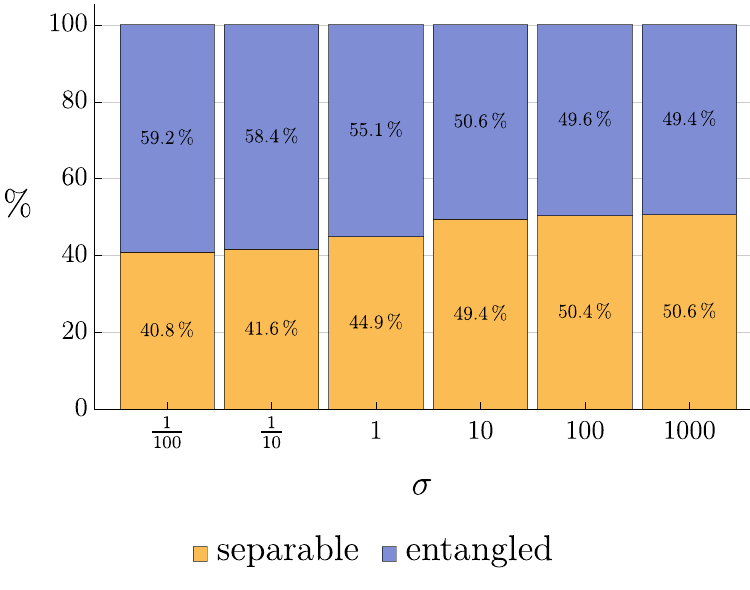} \qquad \includegraphics[width=0.43\textwidth]{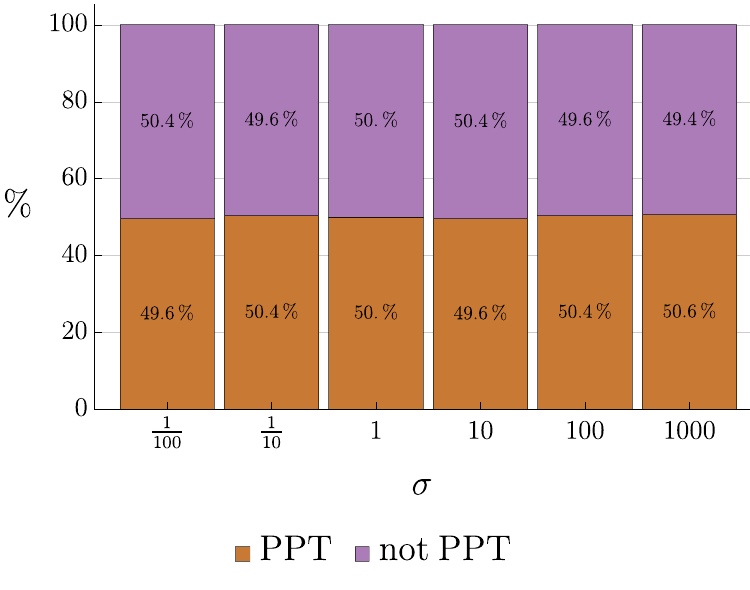} 
    \caption{Numerical histograms in the case of $n=10+10$ modes for the proportion of entangled/separable states (left) and PPT/non-PPT states (right) when $\sigma\in\{\frac{1}{100},\frac{1}{10},1,10,100,1000\}$. The sample size was 10 000 for all $\sigma$.}
    \label{fig:sigma-sep-ppt}
\end{figure}

Lastly, we want to explore how changing the parameter $\sigma$ controlling the spread of the $\GOE$ random matrix affects the entanglement properties of the elements in $\RQCM(2n, \sigma)$. In the previous cases we had fixed $\sigma=1$ but now we will vary $\sigma$ in the case when we have a fixed number of total modes $n=20$ with an even partition $n=10+10$. In particular, we look at different values of the variance parameter $\sigma\in\{\frac{1}{100},\frac{1}{10},1,10,100,1000\}$. The details of the sample sizes and the portions of separable states and PPT states is represented in \cref{fig:sigma-sep-ppt}. For the entangled states we again did the finer analysis of determining for each sample what is the maximum $k$ such that the matrix is $k$-extendible. Our findings can be found in \cref{fig:sigma-extendible}.

From \cref{fig:sigma-sep-ppt} we see that as we increase $\sigma$ also the fraction of separable states increases. This corresponds to the intuition that larger $\sigma$ yields more spread-out GOE matrices $G$ with a larger probability that $H=G$. On the other hand, models with small $\sigma$ would need to be shifted to land in the set of quantum covariance matrices. Since the resulting matrix is on the boundary of this set, the probability that it is separable tends to be smaller, see \cref{fig:closest-QCM}. In a similar vein, \cref{fig:sigma-extendible} shows that those states that are entangled will have a higher level of entanglement as $\sigma$ increases. Also it is noteworthy that as $\sigma$ increases the PPT criterion seems to be better at detecting the separable states.

\begin{figure}[t]
    \centering
    \includegraphics[width=0.4\textwidth]{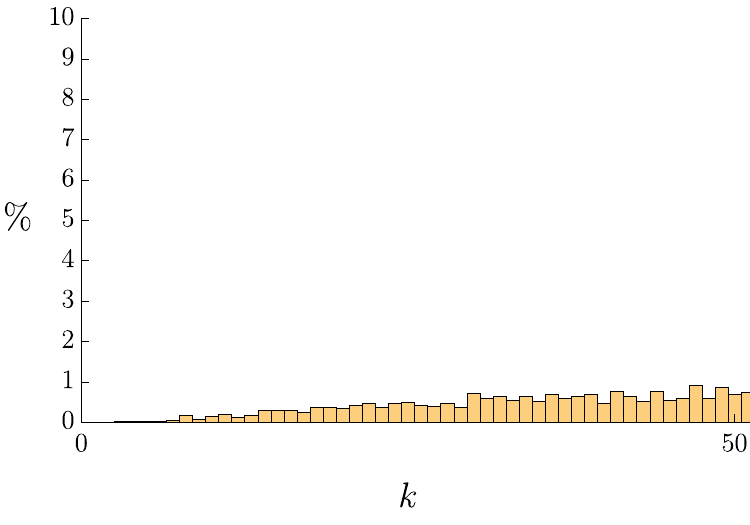} \qquad \includegraphics[width=0.4\textwidth]{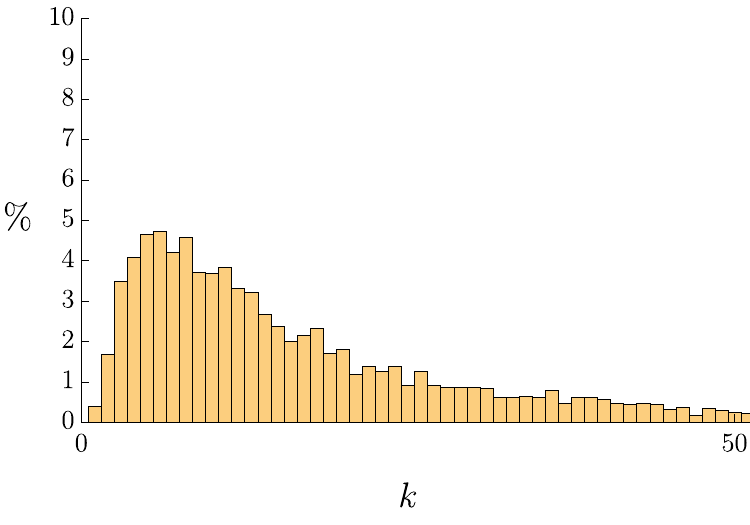} \\
    \includegraphics[width=0.4\textwidth]{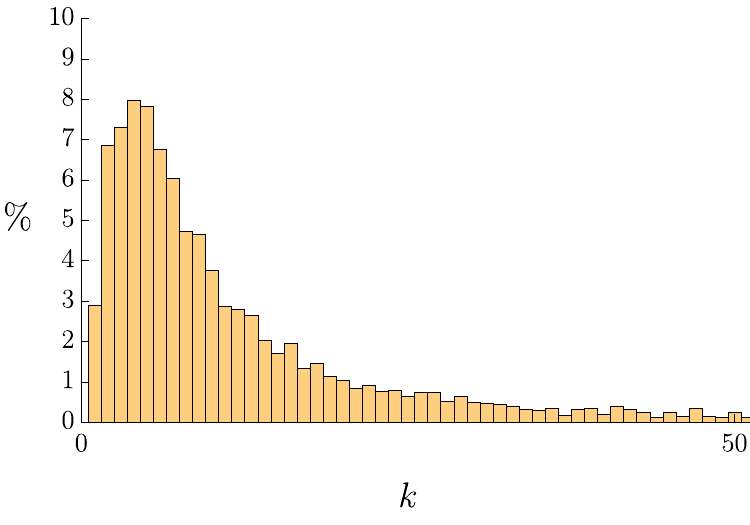} \qquad \includegraphics[width=0.4\textwidth]{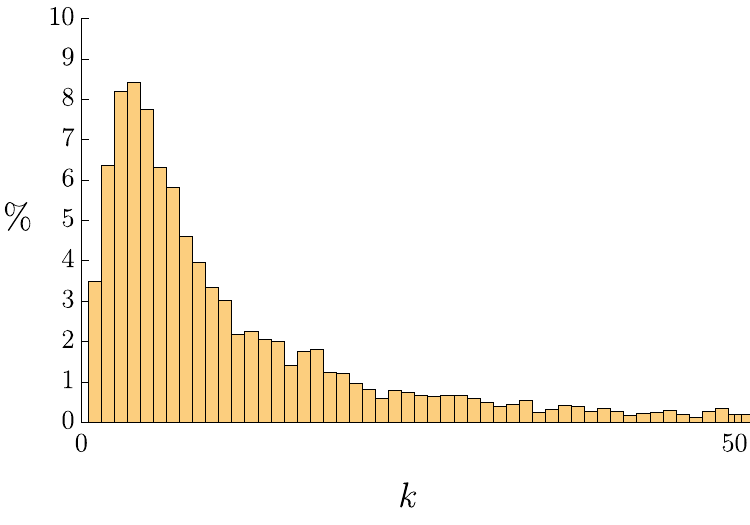}
    \caption{Numerical histograms in the case of $n=10+10$ modes ($\sigma=\frac{1}{10}$ top left, $\sigma=1$ top right, $\sigma=10$ bottom left and $\sigma=100$ bottom right) for the distribution of maximum values of $k$ for each entangled random quantum covariance matrix such that the matrix is $k$-extendable.}
    \label{fig:sigma-extendible}
\end{figure}

\bigskip 

\noindent\textbf{Acknowledgements.} 
LL acknowledges support from the European Union’s Horizon 2020 Research and Innovation Programme under the Programme SASPRO 2 COFUND Marie Sklodowska-Curie grant agreement No. 945478 as well as from projects APVV22-0570 (DeQHOST) and VEGA 2/0183/21 (DESCOM). LL also acknowledges support from the collaboration of the French Embassy in Slovakia, the French Institute in Slovakia and the Slovak Ministry of Education, Science, Research and Sports for funding the research visit to Toulouse where this project was started.
I.N.~was supported by the ANR projects \href{https://esquisses.math.cnrs.fr/}{ESQuisses}, grant number ANR-20-CE47-0014-01 and \href{https://www.math.univ-toulouse.fr/~gcebron/STARS.php}{STARS}, grant number ANR-20-CE40-0008, and by the PHC program \emph{Star} (Applications of random matrix theory and abstract harmonic analysis to quantum information theory).
R.S. acknowledges financial support from \href{https://dst.gov.in}{DST}, Govt. of India, project number  DST/ICPS/QuST/Theme- 2/2019/General Project Q-90, and the French \href{https://www.cnrs.fr/en}{CNRS}, for supporting a visit to Toulouse, where this project was started. 

\bibliographystyle{alpha}
\newcommand{\etalchar}[1]{$^{#1}$}
\def\polhk#1{\setbox0=\hbox{#1}{\ooalign{\hidewidth
  \lower1.5ex\hbox{`}\hidewidth\crcr\unhbox0}}} \def\cprime{$'$}

\end{document}